\documentclass[reqno]{amsart}

\usepackage{amsmath,amsthm,amssymb,amsfonts,enumerate,graphicx,color,pstricks}
\numberwithin{equation}{section} 
\theoremstyle{plain}
\newtheorem{theo+}           {Theorem}      [section]
\newtheorem{prop+}  [theo+]  {Proposition}
\newtheorem{coro+}  [theo+]  {Corollary}
\newtheorem{lemm+}  [theo+]  {Lemma}
\newtheorem{defi+}  [theo+]  {Definition}
\newtheorem{conj+}  [theo+]  {Conjecture}

\theoremstyle{definition}
\newtheorem{rema+}  [theo+]  {Remark}
\newtheorem{prob+}  [theo+]  {Problem}
\newtheorem{exam+}  [theo+]  {Example}

\newenvironment{theorem}{\begin{theo+}}{\end{theo+}}
\newenvironment{proposition}{\begin{prop+}}{\end{prop+}}
\newenvironment{corollary}{\begin{coro+}}{\end{coro+}}
\newenvironment{lemma}{\begin{lemm+}}{\end{lemm+}}

\newcommand{\om}{\omega}

\newcommand{\pfaff}{\mathop{\mathrm{pfaff}}}
\newcommand{\ti}{\mathrm i}
\newcommand{\sgn}{\operatorname{sgn}}

\begin{document}

\baselineskip 18pt
\larger[2]
\title
[Elliptic pfaffians and solvable lattice models] 
{Elliptic pfaffians and solvable lattice models}
\author{Hjalmar Rosengren}
\address
{Department of Mathematical Sciences
\\ Chalmers University of Technology and University of Gothenburg\\SE-412~96 G\"oteborg, Sweden}
\email{hjalmar@chalmers.se}
\urladdr{http://www.math.chalmers.se/{\textasciitilde}hjalmar}

\thanks{Supported by the Swedish Science Research
Council (Vetenskapsr\aa det). 
Some of the research was performed while visiting The Australian National University, Canberra and 
 Simons Center for Geometry and Physics,
Stony Brook University. I gratefully acknowledge support from these two institutions.}

\begin{abstract}
We introduce and study twelve multivariable theta functions defined by pfaffians with elliptic function entries. We show that, when the crossing parameter is a cubic root of unity, the domain wall partition function for the eight-vertex-solid-on-solid model can be written as a sum of two of these pfaffians. As a limit case, we express the domain wall partition function for the three-colour model as a sum of two Hankel determinants.  We also show that certain solutions of the  $TQ$-equation for the supersymmetric eight-vertex model can be expressed in
terms of elliptic pfaffians.
\end{abstract}

\maketitle

\section{Introduction}  

The Izergin--Korepin determinant \cite{ick}
 is an explicit expression for the partition function of the inhomogeneous six-vertex model with domain-wall boundary conditions.
In an appropriate normalization, it takes the form
\begin{equation}\label{ik}\frac{\prod_{i,j=1}^n(u_i-v_j)(qu_i-v_j)}{\prod_{1\leq i<j\leq n}(u_i-u_j)(v_i-v_j)}\,\det_{1\leq i,j\leq n}\left(\frac 1{(u_i-v_j)(qu_i-v_j)}\right),\end{equation}
where $u_1,\dots,u_n$, $v_1,\dots,v_n$ are spectral parameters and $q$ the crossing parameter of the model.
A particularly interesting case is  $q=e^{2\ti\pi/3}$, which was used by Kuperberg in his proof of the alternating sign 
matrix conjecture \cite{ku}.  Okada \cite{o} and Stroganov \cite{st} found that, in this case, 
the partition function is not only separately symmetric in the variables 
$u_j$ and $v_j$, but jointly symmetric.
More precisely, if we replace each $u_j$ by $q u_j$ and multiply with the constant $q^{\binom n2}$, we obtain
the function
\begin{equation}\label{npf}F_n
=\frac 1{\prod_{1\leq i<j\leq n}(u_i-u_j)(v_i-v_j)}\prod_{i,j=1}^n\frac{u_i^3-v_j^3}{u_i-v_j}\det_{1\leq i,j\leq n}\left(\frac{u_i-v_j}{u_i^3-v_j^3}\right),
\end{equation}
which is a symmetric polynomial in $x=(u_1,\dots,u_n,v_1,\dots,v_n)$. In fact, it equals the Schur polynomial
\begin{equation}\label{si}F_n=s_{n-1,n-1,\dots,1,1,0,0}(x). \end{equation}
Sundquist \cite{su} gave an alternative expression for the same Schur polynomial as a pfaffian. Indeed, squaring the variables, 
\begin{equation}\label{pfi}s_{n-1,n-1,\dots,1,1,0,0}(x^2)=\prod_{1\leq i<j\leq 2n}\frac{x_i^{3}+x_j^{3}}{x_i^2-x_j^2}\pfaff_{1\leq i,j\leq 2n}\left(\frac{x_i^2-x_j^2}{x_i^{3}+x_j^{3}}\right).
\end{equation}
In contrast to \eqref{npf}, this expression  displays the symmetry between all $2n$ variables.

The purpose of the present paper is to introduce generalizations of \eqref{pfi} involving elliptic functions and to give applications to solvable lattice models. It turns out that, in some situations when simple determinant formulas such as
\eqref{npf} are not available, pfaffian formulas similar to \eqref{pfi} exist.

In \S \ref{pfds},
 we introduce twelve multivariable theta functions, denoted $P_n^{(\sigma)}$ and defined by pfaffians. In fact, all twelve functions are related by modular transformations, see Proposition~\ref{pmp}. In Proposition \ref{ap}, we show that the functions $P_n^{(\sigma)}$ can be characterized by certain analytic properties. This is
fundamental for recognizing them in the context of solvable models. For such applications, one is particularly interested in the homogeneous limit, when all variables coincide. Another interesting case is the trigonometric limit, when the quasi-period of the theta function becomes infinite.
For analyzing both these limits,  we follow the approach of \cite{rt,rsq}, where we considered the pfaffians
\begin{equation}\label{nupf}\pfaff_{1\leq i,j\leq 2n}\left(\frac{x_j\theta(x_i^2/x_j^2;p^2)}{x_i\theta(px_i^2/x_j^2;p^2)}\right),\qquad \pfaff_{1\leq i,j\leq 2n}\left(\frac{\theta(x_i^2/x_j^2;p^2)}{\theta(-x_i^2/x_j^2;p^2)}\right)  \end{equation}
(see \S \ref{ps} for the notation). Their homogeneous limits contain information on the number of 
representations of an integer as a sum of $4n^2$ triangles and squares, respectively, which led to new proofs and generalizations of 
results in \cite{gm,kw,m,z}. Following the same approach leads to
 expansions of $P_n^{(\sigma)}$ into Schur polynomials and other symmetric functions, see
 Proposition \ref{spp} and Proposition \ref{sqe}. The trigonometric limit is treated in Proposition \ref{tp}. In the homogeneous limit, our pfaffians degenerate to Hankel determinants, see Theorem \ref{hdt}. This may be useful for analyzing the subsequent limit as $n\rightarrow\infty$, which is
of great importance for applications.

In \S \ref{as}, we give applications to statistical mechanics. We first consider the 8VSOS (eight-vertex-solid-on-solid) model \cite{b8}, which has the same states as the six-vertex model  but more general Boltzmann weights. 
In Theorem~\ref{dwt}, we find a new expression for the domain wall partition function as a sum of two pfaffians, valid when
$q=e^{2\ti\pi/3}$. In the homogeneous limit, this leads to
an expression for the domain wall partition function of the three-colour model as a sum of two Hankel determinants, see Corollary \ref{thc}. 
Finally, in Theorem \ref{tqt} we give an application to the eight-vertex model,  at the "supersymmetric" parameter value $\eta=\pi/3$. 
Assuming the conjecture of Razumov and Stroganov \cite{ras0}
  that the transfer matrix on a chain of odd length has a particular eigenvalue, we show that the corresponding  $Q$-operator eigenvalues can again be written
as sums of two elliptic pfaffians. 

\section{Preliminaries}\label{ps}

\subsection{Notation}\label{ns}
 Throughout, $\tau$ is a parameter in the upper half-plane and $p=e^{\ti\pi\tau}$, $p^\lambda=e^{\ti\pi\tau\lambda}$.
We will write $\omega=e^{2\ti\pi/3}$. When $x=(x_1,\dots,x_m)$ is a vector, we write $X=x_1\dotsm x_m$ and
\begin{equation}\label{vdn}\Delta(x)=\prod_{1\leq i<j\leq m}(x_i-x_j).\end{equation}
We will often use the Legendre symbol $(k/3)$, which is the representative of $k\ \operatorname{mod}\ 3$ in $\{-1,0,1\}$.

The pfaffian of a skew-symmetric even-dimensional matrix is given by
$$\pfaff_{1\leq i,j\leq 2n}(a_{ij})=\frac 1{2^nn!}\sum_{\sigma\in S_{2n}}\sgn(\sigma)\prod_{j=1}^na_{\sigma(2j-1)\sigma(2j)}. $$
This can alternatively be viewed as a sum over pairings  (decompositions into disjoint $2$-element sets) on $[1,2n]$, where
$\sigma$ corresponds to the pairing 
$$\{\sigma(1),\sigma(2)\},\dots,\{\sigma(2n-1),\sigma(2n)\}. $$
The prefactor  $1/2^nn!$ is then absent, as each pairing corresponds to $2^n n!$ permutations.

\subsection{Theta functions}

We will  use the notation
\begin{align*}(a;p)_\infty&=\prod_{j=0}^\infty(1-ap^j),\\
\theta(x;p)&=(x;p)_\infty(p/x;p)_\infty. 
\end{align*}
Repeated arguments stand for products, that is,
\begin{align*}
(a_1,\dots,a_m;p)_\infty&=(a_1;p)_\infty\dotsm(a_m;p)_\infty, \\
\theta(x_1,\dots,x_m;p)&=\theta(x_1;p)\dotsm\theta(x_m;p). 
\end{align*}

The theta function satisfies 
\begin{equation}\label{tqp}\theta(px;p)=\theta(x^{-1};p)=-x^{-1}\theta(x;p)
\end{equation}
and the modular transformations
\begin{equation}\label{tmt}e^{-\ti\pi z/(c\tau+d)}\theta(e^{2\ti\pi z/(c\tau+d)};e^{2\pi\ti(a\tau+b)/(c\tau+d)})=Ce^{\ti\pi cz^2/(c\tau+d)-\ti\pi z}\theta(e^{2\ti\pi z};e^{2\pi\ti\tau}), \end{equation}
where $\left[\begin{smallmatrix}a&b\\c&d\end{smallmatrix}\right]\in\mathrm{SL}(2,\mathbb Z)$ and $C$ is a certain multiplier
independent of $z$.

Occasionally, we will use the more classical notation  
\begin{subequations}\label{jt}
\begin{align}
\theta_1(z|\tau)&=\ti p^{1/4}(p^2;p^2)_\infty x^{-1}\theta(x^2;p^2),\\
\theta_2(z|\tau)&=p^{1/4}(p^2;p^2)_\infty x^{-1}\theta(-x^2;p^2),\\
\theta_3(z|\tau)&=(p^2;p^2)_\infty\theta(-px^2;p^2),\\
\theta_4(z|\tau)&=(p^2;p^2)_\infty\theta(px^2;p^2),
\end{align}
\end{subequations}
where 
$x=e^{\ti z}$, $p=e^{\ti\pi \tau}$.

The Laurent expansion of $\theta$ is given by Jacobi's triple product identity
\begin{equation}\label{jtp}\theta(x;p)=\frac{1}{(p;p)_\infty}\sum_{n=-\infty}^\infty(-1)^np^{\binom n2}x^n. \end{equation}
We will need the quintuple product identity \cite{gr,w} in the form
\begin{equation}\label{wqp}(p^2;p^2)_\infty \theta(x,px,-px;p^2)=\sum_{n=-\infty}^\infty
\left(\frac{n+1}3\right)p^{\frac{n(n-1)}3}x^{n};\end{equation}
see  \S \ref{ns} for the definition of the Legendre symbol.
Another useful identity, due to Kronecker \cite{we}
\begin{equation}\label{kr}\frac{(p;p)_\infty^2\theta(ax;p)}{\theta(a,x;p)}=\sum_{n=-\infty}^\infty\frac{x^n}{1-ap^n},\qquad |p|<|x|<1
\end{equation}
can be recognized as a special case of Ramanujan's ${}_1\psi_1$-sum \cite{gr}.

 Finally, we mention the relations
(recall that $\omega=e^{2\ti\pi/3}$)
\begin{equation}\label{ts}\theta(-p\omega;p^2)=\frac 1{\theta(-p;p^6)},\qquad \theta(-\omega;p^2)=-\frac{\omega^2}{\theta(-p^2;p^6)}. \end{equation}
These are easy to prove by manipulating infinite products; for instance,
$$\theta(-p\omega;p^2)=(-p\omega,-p\omega^2;p^2)_\infty=\frac{(-p^3;p^6)_\infty}{(-p;p^2)_\infty}=\frac{1}{(-p,-p^5;p^6)_\infty}=\frac{1}{\theta(-p;p^6)}.$$

\section{Elliptic pfaffians}\label{pfs}

\subsection{The pfaffians $P_n^{(\sigma)}$}
\label{pfds}

We will introduce twelve multivariable theta functions
$$P_n^{(\sigma)}=P_n^{(\sigma)}(z_1,\dots,z_{2n};\tau),$$
labelled by a positive integer $n$ and an element
$$\sigma\in\Sigma=\{0,1,2,3,4,6,\hat 0,\hat 1,\hat 2,\hat 3,\hat 4,\hat 6\}.$$
The choice of labelling
is explained in \S \ref{mts}, where we also show that
the twelve functions are related by modular transformations.

For integer labels, $P_n^{(\sigma)}$ are defined by 
\begin{align*}
P_n^{(0)}&=\prod_{1\leq i<j\leq 2n}\theta(p^{1/3}x_i^2/x_j^2;p^{2/3})
\pfaff_{1\leq i,j\leq 2n}\left(\frac{x_j^2\theta(x_i^2/x_j^2,-x_i^2/x_j^2,px_i^2/x_j^2;p^2)}{x_i^2\theta(p^{1/3}x_i^2/x_j^2;p^{2/3})}\right),\\
P_n^{(1)}&=\prod_{1\leq i<j\leq 2n}\frac{x_j^3}{x_i^3}\,\theta(-x_i^6/x_j^6;p^6)
\pfaff_{1\leq i,j\leq 2n}\left(\frac{x_i\theta(x_i^2/x_j^2,-x_i^2/x_j^2,px_i^2/x_j^2;p^2)}{x_j\theta(-x_i^6/x_j^6;p^6)}\right),\\
 P_n^{(2)}&=\prod_{1\leq i<j\leq 2n}\theta(p^3x_i^6/x_j^6;p^6)
\pfaff_{1\leq i,j\leq 2n}\left(\frac{x_j\theta(x_i^2/x_j^2,px_i^2/x_j^2,-px_i^2/x_j^2;p^2)}{x_i\theta(p^3x_i^6/x_j^6;p^6)}\right),\\
P_n^{(3)}&=\prod_{1\leq i<j\leq 2n}\frac{x_j}{x_i}\,\theta(-x_i^2/x_j^2;p^{2/3})
\pfaff_{1\leq i,j\leq 2n}\left(\frac{x_j\theta(x_i^2/x_j^2,-x_i^2/x_j^2,px_i^2/x_j^2;p^2)}{x_i\theta(-x_i^2/x_j^2;p^{2/3})}\right),\\
P_n^{(4)}&=\prod_{1\leq i<j\leq 2n}\theta(p^3x_i^6/x_j^6;p^6)
\pfaff_{1\leq i,j\leq 2n}\left(\frac{x_j^2\theta(x_i^2/x_j^2,-x_i^2/x_j^2,px_i^2/x_j^2;p^2)}{x_i^2\theta(p^3x_i^6/x_j^6;p^6)}\right),\\
P_n^{(6)}&=\prod_{1\leq i<j\leq 2n}\theta(p^{1/3}x_i^2/x_j^2;p^{2/3})
\pfaff_{1\leq i,j\leq 2n}\left(\frac{x_j\theta(x_i^2/x_j^2,px_i^2/x_j^2,-px_i^2/x_j^2;p^2)}{x_i\theta(p^{1/3}x_i^2/x_j^2;p^{2/3})}\right),\\
\end{align*}
where $x_j=e^{\ti\pi z_j}$ and as always $p^{\lambda}=e^{\ti\pi\tau\lambda}$.
The function $P_n^{(\hat\sigma)}$ is  obtained from $P_n^{(\sigma)}$ by 
replacing all factors of the form
 $\theta(\pm p^\lambda x_i^k/x_j^k;p^{2\lambda})$ 
with  $\theta(\mp p^\lambda x_i^k/x_j^k;p^{2\lambda})$.
Equivalently,
\begin{equation}\label{phs}P_n^{(\hat\sigma)}(z_1,\dots,z_{2n};\tau)=P_n^{(\sigma)}(z_1,\dots,z_{2n};\tau+3). 
\end{equation}

Using that $\theta(x;0)=1-x$, we find that 
\begin{align}
\nonumber\lim_{p\rightarrow 0}P_n^{(1)}&=\prod_{1\leq i<j\leq 2n}\left(\frac{x_j^3}{x_i^3}+\frac{x_i^3}{x_j^3}\right)\pfaff_{1\leq i,j\leq 2n}\left(x_ix_j\frac{x_j^4-x_i^4}{x_j^6+x_i^6}\right)\\
\nonumber
&=\frac {(-1)^n}{X^{6n-4}}\prod_{1\leq i<j\leq 2n}\left(x_j^6+x_i^6\right)\pfaff_{1\leq i,j\leq 2n}\left(\frac{x_i^4-x_j^4}{x_j^6+x_i^6}\right)\\
\label{tsp}&=\frac {(-1)^n}{X^{6n-4}}\,\Delta(x^4)s_{n-1,n-1,\dots,1,1,0,0}(x^2),
\end{align}
where we used Sundquist's identity \eqref{pfi} and the notation $X=x_1\dotsm x_{2n}$.  Thus, as the functions $P_n^{(\sigma)}$ are all related to $P_n^{(1)}$ by modular transformations, we can think of them as elliptic extensions of the pfaffian in \eqref{pfi}. 
However, these transformations do not in general preserve the point $p=0$, so the trigonometric limits
of $P_n^{(\sigma)}$ and $P_n^{(1)}$ may be different; see
Proposition \ref{tp} for details.

In classical notation  \eqref{jt}, the functions  $P_n^{(\sigma)}$ are given by a factor
independent of the variables $z_j$ times
\begin{subequations}
\begin{multline}\label{pca}\prod_{1\leq i<j\leq 2n}\theta_k(3w_i-3w_j|3\tau)\\
\times \pfaff_{1\leq i,j\leq 2n}\left(\frac{\theta_1(w_i-w_j|\tau)\theta_k(w_i-w_j|\tau)\theta_l(w_i-w_j|\tau)}{\theta_k(3w_i-3w_j|3\tau)}\right)\end{multline}
or
\begin{multline}\label{pcb}\prod_{1\leq i<j\leq 2n}\theta_k(w_i-w_j|\tau/3) \\
\times\pfaff_{1\leq i,j\leq 2n}\left(\frac{\theta_1(w_i-w_j|\tau)\theta_k(w_i-w_j|\tau)\theta_l(w_i-w_j|\tau)}{\theta_k(w_i-w_j|\tau/3)}\right),\end{multline}
\end{subequations}
where $(k,l)$ are any two distinct elements in $\{2,3,4\}$
and  $w_j=\pi z_j$.
In the following table, we list the label $\sigma$ corresponding to each pair $(k,l)$:\\
\begin{center}\begin{tabular}{cccccccc} 
&$(2,3)$&$(2,4)$&$(3,2)$&$(3,4)$&$(4,2)$&$(4,3)$&\\
\eqref{pca}&$\hat 1$&$ 1$&$ \hat 4$&$\hat 2$&$4$&$ 2$&\\
\eqref{pcb}&$\hat 3$&$ 3$&$\hat 0$&$\hat 6$&$0$&$6$&$\cdot$
\end{tabular}\end{center}

\subsection{Modular transformations}\label{mts}

By the following result, any two of the functions $P_n^{(\sigma)}$ are related by modular transformations.

\begin{proposition}\label{pmp}
Let $\left[\begin{smallmatrix}a&b\\c&d\end{smallmatrix}\right]\in\mathrm{SL}(2,\mathbb Z)$  be such that $cd$ is divisible by $3$. Then,
\begin{multline}\label{pmi}P_n^{(1)}\left(\frac{z_1}{c\tau+d},\dots,\frac{z_{2n}}{c\tau+d};\frac{a\tau+b}{c\tau+d}\right)\\
\sim \exp\left(\frac{3\ti\pi c}{c\tau+d}\sum_{1\leq i<j\leq 2n}(z_i-z_j)^2\right)F(z_1,\dots,z_{2n}; \tau), \end{multline}
where, in the case $c\equiv 0 \ \operatorname{mod}\ 3$, 
$$F=\begin{cases}
P_n^{(1)}, & (a,b,c,d)\equiv(1,0,0,1)\ \operatorname{mod}\ 2,\\
 P_n^{(\hat1)}, & (a,b,c,d)\equiv(1,1,0,1)\ \operatorname{mod}\ 2,\\
P_n^{(2)}, & (a,b,c,d)\equiv(1,1,1,0)\ \operatorname{mod}\ 2,\\
 P_n^{(\hat2)}, & (a,b,c,d)\equiv(1,0,1,1)\ \operatorname{mod}\ 2,\\
P_n^{(4)}, & (a,b,c,d)\equiv(0,1,1,0)\ \operatorname{mod}\ 2,\\
 P_n^{(\hat 4)}, & (a,b,c,d)\equiv(0,1,1,1)\ \operatorname{mod}\ 2
\end{cases} $$
and, in the case $d\equiv 0 \ \operatorname{mod}\ 3$, 
$$F=\begin{cases}
P_n^{(3)}, & (a,b,c,d)\equiv(1,0,0,1)\ \operatorname{mod}\ 2,\\
 P_n^{(\hat 3)}, & (a,b,c,d)\equiv(1,1,0,1)\ \operatorname{mod}\ 2,\\
P_n^{(6)}, & (a,b,c,d)\equiv(1,1,1,0)\ \operatorname{mod}\ 2,\\
 P_n^{(\hat 6)}, & (a,b,c,d)\equiv(1,0,1,1)\ \operatorname{mod}\ 2,\\
P_n^{(0)}, & (a,b,c,d)\equiv(0,1,1,0)\ \operatorname{mod}\ 2,\\
 P_n^{(\hat 0)}, & (a,b,c,d)\equiv(0,1,1,1)\ \operatorname{mod}\ 2.
\end{cases} $$
Here, $\sim$ denotes equality up to a factor
independent of the variables $z_j$.
\end{proposition}

The constant of proportionality in \eqref{pmi} can be expressed in terms of Dedekind sums, but we will not do so here.
Note that, since $ad-bc=1$, there are exactly six possibilities for $(a,b,c,d)\ \operatorname{mod}\ 2$.

\begin{proof} Consider variables related by
$$x=e^{\ti\pi z},\qquad p=e^{\ti\pi\tau},\qquad \tilde x=e^{\ti\pi z/(c\tau+d)},\qquad \tilde p=e^{\ti\pi(a\tau+b)/(c\tau+d)}. $$
 Taking the product of 
\eqref{tmt} and the two identities obtain from \eqref{tmt} after  substituting $z\mapsto z+(c\tau+d)/2$ and $z\mapsto (a\tau+b)/2$ gives
\begin{align*}\tilde x^{-3}\theta(\tilde x^2,-\tilde x^2,\tilde p\tilde x^2;p^2)
&\sim \exp\left(\frac{3\ti\pi cz}{c\tau+d}\left(z+ a\tau+b+c\tau+d\right)\right)\\
&\quad\times x^{-3}\theta\left(x^2,(-1)^dp^cx^2,(-1)^bp^ax^2;p^2\right).\end{align*}
 By \eqref{tqp},
$$\theta\left((-1)^dp^c x^2;p^2\right)\sim \begin{cases} x^{-c}\theta(- x^2;p^2),& c \text{ even (hence } d \text{ odd)},\\
 x^{1-c}\theta\left((-1)^dp x^2;p^2\right), & c \text{ odd}.\end{cases}$$
 Treating the factor $\theta((-1)^bp^ax^2;p^2)$ in the same way, we find that
\begin{equation}\label{nmt}\tilde x^{-2}\theta(\tilde x^2,-\tilde x^2,p\tilde x^2;\tilde p^2)\sim
e^{\frac{3\ti\pi cz^2}{c\tau+d}}\times\begin{cases}x^{-2}\theta(x^2,-x^2,-px^2;p^2), & b \text{\ and } d \text{ odd},\\
x^{-1}\theta(x^2,px^2,-px^2;p^2),& a \text{ and } c \text{ odd},\\
x^{-2}\theta(x^2,-x^2,px^2;p^2),& \text{else}.
\end{cases}
 \end{equation}

If $c\equiv 0\ \operatorname{mod}\ 3$, substituting $(a,b,c,d)\mapsto(a,3b,c/3,d)$,
$\tau\mapsto 3\tau$ and $z\mapsto 3z+(c\tau+d)/2$ in \eqref{tmt} gives
\begin{equation}\label{dmt}\tilde x^{-3}\theta(-\tilde x^6;\tilde p^6)\sim e^{\frac{3\ti\pi cz^2}{c\tau+d}}\times\begin{cases}x^{-3}\theta(-x^6;p^6), & c \text{ even},\\
\theta(p^3x^6;p^6), & d \text{ even},\\
\theta(-p^3x^6;p^6), & \text{else}.\end{cases}
 \end{equation}
Similarly, if $d\equiv 0\ \operatorname{mod}\ 3$, substituting $(a,b,c,d)\mapsto(3a,b,c,d/3)$,
$\tau\mapsto \tau/3$ and $z\mapsto z+(c\tau+d)/6$ in \eqref{tmt} gives
\begin{equation}\label{emt}\tilde x^{-3}\theta(-\tilde x^6;\tilde p^6)\sim e^{\frac{3\ti\pi cz^2}{c\tau+d}}\times\begin{cases}x^{-1}\theta(-x^2;p^{2/3}), & c \text{ even},\\
\theta(p^{1/3}x^2;p^{2/3}), & d \text{ even},\\
\theta(-p^{1/3}x^2;p^{2/3}), & \text{else}.\end{cases}
 \end{equation}
Combining the identities \eqref{nmt}, \eqref{dmt} and \eqref{emt} we arrive at \eqref{pmi}.
\end{proof}

One can obtain a more conceptual understanding of Proposition~\ref{pmp} as follows.
Recall the notation
\begin{align*}
\Gamma_0(n)&=\left\{\left[\begin{matrix}a& b\\ c& d\end{matrix}\right]\in\mathrm{SL}(2,\mathbb Z);\, n\mid c\right\},\\
\Gamma_0(m,n)&=\left\{\left[\begin{matrix}a& b\\ c& d\end{matrix}\right]\in\mathrm{SL}(2,\mathbb Z);\,m\mid b,\ n\mid c\right\}.
\end{align*}
The first case of Proposition \ref{pmp} states that
 $P_n^{(1)}$ is essentially invariant under
the group $\Gamma=\Gamma_0(2,6)\simeq \Gamma_0(12)$.
It follows from the results of \cite{al} (see also \cite{rsa}) that  the normalizer $\mathrm{N}$ of $\Gamma$ in $\mathrm{SL}(2,\mathbb R)$
is  a disjoint union  $\mathrm N=\mathrm N_1\cup \mathrm N_2$, where $\mathrm N_1=\Gamma_0(3)$ and
$$\mathrm N_2=\left[\begin{matrix}0& 1/\sqrt 3\\ -\sqrt 3&0\end{matrix}\right]\Gamma_0(3)=\left\{\left[\begin{matrix}a\sqrt 3& b/\sqrt 3\\ c\sqrt 3 & d/\sqrt 3\end{matrix}\right];\,\left[\begin{matrix}a& b\\ c& d\end{matrix}\right]\in\mathrm{SL}(2,\mathbb Z), \ 3\mid d\right\}. $$
Consider the left-hand side of \eqref{pmi} for
 $A=\left[\begin{smallmatrix}a&b\\c&d\end{smallmatrix}\right]\in \mathrm N$. 
As  $\mathrm N/\Gamma\simeq S_2\times S_3$
there are  twelve cases to consider. If
$A\in \mathrm N_1$, we are in one of the first six cases of Proposition~\ref{pmp}, which correspond to the functions
\eqref{pca}.
If $A\in \mathrm N_2$, we need to consider
\begin{multline*}
P_n^{(1)}\left(\frac{z_1}{c\sqrt 3 \tau+\frac{d}{\sqrt 3}},\dots,\frac{ z_{2n}}{c\sqrt 3\tau+\frac{d}{\sqrt 3}};\frac{a\sqrt 3\tau+\frac{b}{\sqrt 3}}{c\sqrt 3\tau+\frac d{\sqrt 3}}\right)\\
=P_n^{(1)}\left(\frac{\sqrt 3 z_1}{3c\tau+d},\dots,\frac{\sqrt 3 z_{2n}}{3c\tau+d};\frac{3a\tau+b}{3c\tau+d}\right),\end{multline*}
where $3\mid d$. By Proposition \ref{pmp}, this
is an elementary factor times
$$P_n^{(\sigma)}(\sqrt 3 z_1,\dots,\sqrt 3z_n;3\tau), $$
 where $P_n^{(\sigma)}$ is one of the functions \eqref{pcb}.

Finally, we note that $\Gamma$ has six cusps, which correspond to particular limits of $\tau$ in $\mathbb Q\cup\{\infty\}$. 
 Explicitly, they are given by $\{C_0,C_1,C_2,C_3,C_4,C_6\}$, where
\begin{align*}C_\sigma&=\left\{\frac{k}{6l};\ (k,l)=1,\ k\equiv\pm \sigma\ \operatorname{mod}\ 12\right\},\qquad \sigma\neq 1, \\
C_1&=\left\{\frac{k}{6l};\ (k,l)=1,\ k\equiv\pm 1\ \operatorname{mod}\ 6\right\}\cup\{\infty\}. \end{align*}
The group $\mathrm N/\Gamma$ acts transitively on the cusps.
An element  $\left[\begin{smallmatrix}a&b\\c&d\end{smallmatrix}\right]\in\mathrm{SL}(2,\mathbb Z)$
 maps $\infty$ to $a/c$. It is easy to see that if $c\equiv 0\ \operatorname{mod}\ 3$, then $a/c\in C_1$ if 
$c$ is even, $a/c\in C_4$ if $a$ is even and otherwise $a/c\in C_2$.
If $d\equiv 0\ \operatorname{mod}\ 3$, we have instead $a/c\in C_3,\,C_6,\, C_0$, respectively. Thus, we have labelled the functions $P_n^{(\sigma)}$ so that the behaviour of $P_n^{(1)}$ as $p\rightarrow e^{\ti\pi\sigma/6}$ is related to the behaviour 
of $P_n^{(\sigma )}$ (and $P_n^{(\hat\sigma)}$) as $p\rightarrow 0$.

\subsection{Characteristic properties}

We will now show that the functions $P_n^{(\sigma)}$ are  characterized by certain analytic properties. This  will  be used
in \S \ref{as} to identify some of them with quantities appearing in statistical mechanics.

\begin{lemma}\label{tql}
Let $f(x)$ denote one of the three functions
$$x^{-1}\theta(x^2,px^2,-px^2;p^2),\qquad
x^{-2}\theta(x^2,-x^2,px^2;p^2),\qquad x^{-2}\theta(x^2,-x^2,-px^2;p^2). $$
Then,
\begin{subequations}\label{csf}
\begin{equation}\label{csa}f(\om^{-1}x)+f(x)+f(\om x)=0,\qquad \om=e^{2\ti\pi/3}, \end{equation}
\begin{equation}\label{csb}x^{-4}f(p^{-2/3}x)+p^{-4/3}f(x)+x^4f(p^{2/3}x)=0. \end{equation}
\end{subequations}
\end{lemma}

\begin{proof}
For the first choice of $f$, the quintuple product identity
\eqref{wqp} gives
$$ f(x)=\frac1{(p^2;p^2)_\infty}\sum_{n=-\infty}^\infty\left(\frac{n+1}3\right)p^{\frac{n(n-1)}3}x^{2n-1}. $$
The identity \eqref{csa} then follows from the fact that $\omega^{1-2n}+1+\omega^{2n-1}=0$ unless $n\equiv-1\ \operatorname{mod}\ 3$,
that is, $((n+1)/3)=0$. As for \eqref{csb}, appropriate shifts of the summation index give
\begin{multline*}
x^{-4}f(p^{-2/3}x)+p^{-4/3}f(x)+x^4f(p^{2/3}x)\\
=\frac1{(p^2;p^2)_\infty}\sum_{n=-\infty}^\infty\left(\left(\frac{n+3}3\right)+\left(\frac{n+1}3\right)+\left(\frac{n-1}3\right)\right)p^{\frac{n^2-n-4}3}x^{2n-1}=0.
\end{multline*}
Finally,  we note that if $f$ solves \eqref{csf} then so do $x\mapsto x^3f(\pm \sqrt{p}x)$.  This yields the other two solutions.
\end{proof}

\begin{lemma}\label{nvl}
The functions $P_n^{(\sigma)}$ never vanish identically.
\end{lemma}

\begin{proof}
We specialize the variables in $P_n^{(\sigma)}$ as  $z_{2n-1}-z_{2n}= \gamma$, where
$$\gamma=\begin{cases}  \tau/6,& \sigma=0,6,\\
1/2+ \tau/6, & \sigma=\hat 0,\hat 6,\\
1/2+ \tau/3, & \sigma=3,\hat 3,\\
 1/6, & \sigma=1,\hat 1,\\
 1/6+\tau/2, & \sigma=\hat 2,\hat 4,\\
 1/3+\tau/2, & \sigma=2,4.
\end{cases} $$
 Writing 
$$P_n^{(\sigma)}=\prod_{1\leq i<j\leq 2n}b_{ij}\pfaff_{1\leq i,j\leq 2n}\left(\frac{a_{ij}}{b_{ij}}\right), $$
we then have $b_{2n-1,2n}=0$, $a_{2n-1,2n}\neq 0$.
Thus, 
 expanding the pfaffian, only terms containing $a_{2n-1,2n}$ contributes, which gives
\begin{multline}\label{psr}P_n^{(\sigma)}(z_1,\dots,z_{2n-2},z_{2n}+\gamma,z_{2n};\tau)\\
=a_{2n-1,2n}\prod_{j=1}^{2n-2}b_{j,2n-1}b_{j,2n}\,P_{n-1}^{(\sigma)}(z_1,\dots,z_{2n-2};\tau).\end{multline}
As the prefactor is non-zero for generic values of $z_1,\dots,z_{2n-2}$,
the result follows by induction on $n$.
\end{proof}

\begin{proposition}\label{ap}
Let $n$ be a positive integer,
 $\sigma\in\Sigma$ and let $\tau$ be in the upper half-plane. Let $z_2,\dots,z_{2n}$ be generic complex numbers 
 and let 
 $t=e^{2\ti\pi(z_2+\dots+z_{2n})}$.
Suppose $f$ is an entire function   such that
\begin{align}
\label{eo} f(z+1)&=\begin{cases} f(z), & \sigma=0,\hat 0,1,\hat 1,3,\hat 3,4,\hat 4,\\
 -f(z), & \sigma=2,\hat 2,6,\hat 6,\end{cases}\\
\label{fqp} f(z+\tau)&=\frac{t^3}{e^{(6n-3)\ti\pi (2z+\tau)}}\times\begin{cases}f(z),&\sigma=0,1,2,\hat 2,3,4,6,\hat 6,\\
-f(z),& \sigma=\hat 0,\hat 1,\hat 3,\hat 4,\end{cases}\end{align}
\begin{equation}
\label{ptq} f\left(z-\frac{2}3\right)+f(z)+f\left(z+\frac 23\right)=0, \qquad \sigma=1,\hat 1,2,\hat 2,4,\hat 4,\end{equation}
\begin{multline*} t^2e^{-(8n-4)\ti\pi z}f\left(z-\frac{2\tau}3\right)+e^{-\frac{(8n-4)\ti\pi\tau}3}f(z)
+t^{-2}e^{(8n-4)\ti\pi z}f\left(z+\frac{2\tau}3\right)=0,\\
  \sigma=0,\hat 0,3,\hat 3, 6,\hat 6,\end{multline*}
\begin{align}
\label{vc} f(z_2)&=\dots=f(z_{2n})=0, \qquad \text{for all }\sigma, \\
\nonumber  f\left(z_2+\frac 12\right)&=\dots=f\left(z_{2n}+\frac 12\right)=0,\qquad \sigma=1,\hat 1,3,\hat 3,\\
\label{ve} f\left(z_2+\frac\tau 2\right)&=\dots=f\left(z_{2n}+\frac\tau 2\right)=0,\qquad \sigma=0,2,4,6,\\
\nonumber f\left(z_2+\frac{1+\tau} 2\right)&=\dots=f\left(z_{2n}+\frac{1+\tau} 2\right)=0,\qquad \sigma=\hat 0,\hat 2,\hat 4,\hat 6.
\end{align}
Then, 
$$f(z)=CP_n^{(\sigma)}(z,z_2,\dots,z_{2n};\tau), $$
with $C$ independent of $z$.
\end{proposition}

\begin{proof} 
All cases can be proved by the same method or deduced from each other using Proposition \ref{pmp}. Thus,
we consider only
the case $\sigma=1$.  Let $W$  be the space of entire 
functions satisfying all the stated conditions (for $\sigma=1$)
and let $V$ be the larger space where we do not assume \eqref{vc}.
Let $g(z)=P_n^{(1)}(z,z_2,\dots,z_{2n};\tau)$. By Lemma \ref{nvl}, $g$ is not identically zero for generic values of $z_j$. 
Thus, it suffices to show that $g\in W$ 
and that $\dim W=1$.

For $2\leq k\leq 2n$, let
$$f_k(z)=x_k^2x^{-2}\theta(x^2/x_k^2,-x^2/x_k^2,px^2/x_k^2;p^2)\prod_{j=2,\,j\neq k}^{2n}x_j^3x^{-3}\theta(-x^6/x_j^6;p^6), $$
where
$x=e^{\ti\pi z}$, $x_j=e^{\ti\pi z_j}$.
By definition, $g$ is a linear combination of the functions $f_k$. We will prove that $(f_k)_{k=2}^{2n}$ is a basis for $V$.
To prove that $g\in W$ it then suffices to
show that $g$ vanishes at the points $z_2,\dots,z_{2n}$, which is clear since
 $P_n^{(1)}$ is anti-symmetric.

That $f_k\in V$ follows easily from
\eqref{tqp} and Lemma \ref{tql}. Since $z_j$ are generic, $f_k(z_j+1/6)=0$ 
if and only if $k\neq j$. This shows that the functions $f_k$ are linearly independent. To prove that they span $V$  we will
prove that any $h\in V$ can be expanded as
\begin{equation}\label{he}h(z)=\sum_{k=2}^{2n}\frac{h(z_k+1/6)}{f_k(z_k+1/6)}\,f_k(z). \end{equation}
To this end,
let  $f(z)$ be the difference of the two sides 
in \eqref{he}. Then, $f\in V$ and $f$ vanishes at the points $z_k+1/2$ and $z_k+1/6$. Using \eqref{eo}, \eqref{fqp} and \eqref{ptq} we deduce that $f$ vanishes at $z_k+(2\mathbb Z+1)/6+\tau\mathbb Z$. Thus,
$$k(z)=\frac{f(z)}{\prod_{k=2}^{2n}x_k^3x^{-3}\theta(-x^6/x_k^6;p^6)} $$
is an entire function. Using again \eqref{eo} and \eqref{fqp},
one checks that $k$ has periods $1$ and $\tau$ so, by Liouville's theorem, it is constant. Since $k$ satisfies \eqref{ptq}, that constant is zero.
This completes the proof that $(f_k)_{k=2}^{2n}$ is a basis for $V$.

It remains to   prove that $\dim(W)=1$. Identifying the elements in $V$ with
their coordinates in the basis $(f_k)_{k=2}^{2n}$, the subspace $W$ 
is identified with the nullspace of the matrix  $(f_j(z_i))_{i,j=2}^{2n}$.
We must show that this matrix has rank $2n-2$. If the
 rank were $2n-1$, then we
would have
\begin{align*}0&\neq\det_{2\leq i,j\leq 2n}(f_j(z_i))\\
&=\prod_{i,j=2}^{2n}x_i^{-3}x_j^3\theta(-x_i^6/x_j^6;p^6)\det_{2\leq i,j\leq 2n}\left(\frac{x_i\theta(x_i^2/x_j^2,-x_i^2/x_j^2,px_i^2/x_j^2;p^2)}{x_j\theta(-x_i^6/x_j^6;p^6)}\right). \end{align*}
As the final matrix is odd-dimensional and skew-symmetric, this is impossible.
If, on the other hand, the rank were less than $2n-2$, then
we would have
$$\det_{2\leq i,j\leq 2n-1}\left(\frac{x_i\theta(x_i^2/x_j^2,-x_i^2/x_j^2,px_i^2/x_j^2;p^2)}{x_j\theta(-x_i^6/x_j^6;p^6)}\right)=0. $$
 This is equivalent to the pfaffian (which is a square root of the determinant)
 being zero, that is,
 $P_{n-1}^{(1)}(z_2,\dots,z_{2n-1};\tau)=0$. Since $z_j$ are generic, this contradicts Lemma 
\ref{nvl}.
\end{proof}

As a non-trivial consequence of Proposition \ref{ap}, 
we obtain the following fact.

\begin{corollary}\label{hsc}
One has
\begin{align*}P_n^{(\sigma)}\left(z_1+\frac12,z_2,\dots,z_{2n};\tau\right)&=A_\sigma B_\sigma^{n-1}P_n^{(\rho)}(z_1,z_2,\dots,z_{2n};\tau),\\
&\hspace*{4cm} \{\sigma,\rho\}=\{1,\hat 1\},\, \{3,\hat 3\},\\
P_n^{(\sigma)}\left(z_1+\frac\tau2,z_2,\dots,z_{2n};\tau\right)&=A_\sigma B_\sigma^{n-1}
\frac{x_2^3\dotsm x_{2n}^3}{x_1^{6n-3}}P_n^{(\rho)}(z_1,z_2,\dots,z_{2n};\tau), \\ 
&\hspace*{4cm} \{\sigma,\rho\}=\{0,6\},\,\{2,4\},\\
P_n^{(\sigma)}\left(z_1+\frac{\tau+1}2,z_2,\dots,z_{2n};\tau\right)&=A_\sigma B_\sigma^{n-1}\frac{x_2^3\dotsm x_{2n}^3}{x_1^{6n-3}}P_n^{(\rho)}(z_1,z_2,\dots,z_{2n};\tau),\\
&\hspace*{4cm} \{\sigma,\rho\}=\{\hat 0,\hat 6\},\,\{\hat 2,\hat 4\},
\end{align*}
where 
\begin{align*}A_\sigma&=\begin{cases}
-1, & \sigma= 1,\hat 1,3,\hat 3,\\
-p^{-1}, & \sigma =0,4,\\
p^{-1}, & \sigma=\hat 0,\hat 4,\\
-p^{-1/2}, & \sigma=2,6,\\
\ti p^{-1/2}, & \sigma =\hat 2,\hat 6,
\end{cases}\\
B_\sigma&=\begin{cases}
-p^{-5/3}\theta(-p^{1/3};p^2)/\theta(-p^{2/3};p^2),& \sigma=0,\\
p^{-5/3}\theta(p^{1/3};p^2)/\theta(-p^{2/3};p^2), & \sigma=\hat 0,\\
-\theta(p;p^6)/\theta(-p;p^6), & \sigma=1,\\
-\theta(-p;p^6)/\theta(p;p^6), & \sigma=\hat 1,\\
p^{-1}\theta(-p;p^6)/\theta(-p^2;p^6), & \sigma=2,\\
-p^{-1}\theta(p;p^6)/\theta(-p^2;p^6),& \sigma=\hat 2,\\
\theta(-p^{1/3};p^2)/\theta(p^{1/3};p^2),& \sigma=3,\\
\theta(p^{1/3};p^2)/\theta(-p^{1/3};p^2),& \sigma=\hat 3,\\
p^{-2}\theta(-p^2;p^6)/\theta(-p;p^6), & \sigma=4,\\
p^{-2}\theta(-p^2;p^6)/\theta(p;p^6), & \sigma=\hat 4,\\
-p^{-4/3}\theta(-p^{2/3};p^2)/\theta(-p^{1/3};p^2), & \sigma=6,\\
-p^{-4/3}\theta(-p^{2/3};p^2)/\theta(p^{1/3};p^2), & \sigma=\hat 6.
\end{cases}
\end{align*}
\end{corollary}

\begin{proof}
As all  cases are similar, we again focus on the case $\sigma=1$. 
We first show  that
\begin{equation}\label{php}P_n^{(1)}(z_1+1/2,z_2,\dots,z_{2n};\tau)=CP_n^{(\hat 1)}(z_1,\dots,z_{2n};\tau), \end{equation}
where $C$ is independent of the variables $z_j$. 
Let $W_\sigma$ denote the space of functions satisfying all conditions
of Proposition \ref{ap}. It is easy to check that if
$f\in W_1$, then $z\mapsto f(z+1/2)$ is in $W_{\hat 1}$. Thus, Proposition \ref{ap} implies that
\eqref{php} holds with
$C=C(z_2,z_3,\dots,z_{2n})$ independent of $z_1$. Interchanging $z_1$ and $z_2$ gives
$$ P_n^{(1)}(z_1,z_2+1/2,\dots,z_{2n};\tau)=C(z_1,z_3,\dots,z_{2n})P_n^{(\hat 1)}(z_1,\dots,z_{2n};\tau).$$
We now observe that
replacing $z_2$ by $z_2+1/2$ 
maps $W_1$ to $W_{\hat 1}$. Thus, $C$ is independent of 
its first variable and, by symmetry, on all variables.

To compute $C$, we specialize $z_{2n-1}=z_{2n}+1/6$  in
\eqref{php}. Then we can apply \eqref{psr} to both sides.
The factors $b_{jk}$ with $j\neq 1$ cancel. 
This gives a recursion for $C=C_n$ of the form
$$\frac{C_n}{C_{n-1}}=\frac{a_{2n-1,2n}b_{1,2n-1}b_{1,2n}}{\hat a_{2n-1,2n}\hat b_{1,2n-1}\hat b_{1,2n}}, $$
where the factors in the denominator come from the right-hand side
of \eqref{php}. Explicitly,
$$\frac{a_{2n-1,2n}}{\hat a_{2n-1,2n}}=\frac{\theta(px_{2n-1}^2/x_{2n}^2;p^2)}{\theta(-x_{2n-1}^2/x_{2n}^2;p^2)}\Bigg|_{x_{2n-1}=e^{\ti\pi/6}x_{2n}}=\frac{\theta(-p\om^2;p^2)}{\theta(p\om^2;p^2)}
=\frac{\theta(p;p^6)}{\theta(-p;p^6)}
, $$
where $\omega=e^{2\ti\pi/3}$ and we used \eqref{ts} (also with $p$ replaced by $-p$) in the last step.
The remaining factors simplify as
$$\frac{b_{1,2n-1}b_{1,2n}}{\hat b_{1,2n-1}\hat b_{1,2n}}
=\frac{\tilde x_1^{-6}\theta(-\tilde x_1^6/x_{2n-1}^6,-\tilde x_1^6/x_{2n}^6;p^6)}{x_1^{-6}\theta(-x_1^6/x_{2n-1}^6,- x_1^6/x_{2n}^6;p^6)}\Bigg|_{\tilde x_1=\ti x_1,\ x_{2n-1}=e^{\ti\pi/6}x_{2n}}=-1.
 $$
We conclude that $C_n=B_1C_{n-1}$.
It remains to check the initial value $C_1=-1$, which is a trivial
computation. 
\end{proof}

\subsection{Laurent expansions}

As we are particularly interested in the limit of the pfaffians $P_n^{(\sigma)}$ when all the variables $x_j$
coincide, it is natural to expand each matrix element as a Laurent series
in $x_i/x_j$, convergent near  $x_i/x_j=1$. This is possible  except
 in the cases $\sigma= 1,\,\hat 1$, as the matrix elements then have poles at certain roots of unity.
We will circumvent this problem by subtracting a rational function containing the problematic poles, see \eqref{sri}. 
A similar strategy was used in \cite{rsq} to study  the second pfaffian in \eqref{nupf}, which is related
to sums of squares.

\begin{lemma}\label{ll}
Let 
\begin{subequations}\label{ci}
\begin{align}
C_0&=-\frac{(p^4;p^4)_\infty}{(p^2;p^2)_\infty^2(p^{4/3};p^{4/3})_\infty},&
C_1&=\frac{(p;p)_\infty}{(p^2;p^2)_\infty^2(p^3;p^3)_\infty},\\
C_2&=\frac{(-p;-p)_\infty}{(p^2;p^2)_\infty^2(-p^3;-p^3)_\infty},&
C_3&=\frac{(p;p)_\infty}{(p^2;p^2)_\infty^2(p^{1/3};p^{1/3})_\infty},\\
C_4&=-\frac{(p^4;p^4)_\infty}{(p^2;p^2)_\infty^2(p^{12};p^{12})_\infty},&
C_6&=-\frac{(-p;-p)_\infty}{(p^2;p^2)_\infty^2(-p^{1/3};-p^{1/3})_\infty}.
\end{align}
\end{subequations}
 Then, for $|p^{1/6}|<|x|<|p^{-1/6}|$,
\begin{align*}
\frac{x^{-2}\theta(x^2,-x^2,px^2;p^2)}{\theta(p^{1/3}x^2;p^{2/3})_\infty}&=C_0\sum_{k=-\infty}^\infty
\frac{1-p^{4k/3}}{1+p^{2k}}\,p^{(k-1)/3}x^{2k},\\
 \frac{x^{-1}\theta(x^2,px^2,-px^2;p^2)}{\theta(p^{1/3}x^2;p^{2/3})_\infty}&=C_6\sum_{k=-\infty}^\infty
\frac{1-p^{(4k-2)/3}}{1+p^{2k-1}}\,p^{(k-1)/3}x^{2k-1}.
\end{align*}
For $|p|<|x|<|p^{-1}|$,
\begin{multline}\label{sri}
\frac{x\theta(x^2,-x^2,px^2;p^2)}{\theta(-x^6;p^6)_\infty}\\
=C_1\left(\frac{x^{-2}-x^2}{x^{-3}+x^3}+\sum_{k=1}^\infty\left(\frac{k+1}3\right)
\frac{(-1)^{k}p^{2k-1}(x^{2k-1}-x^{-2k+1})}{1-p^{2k-1}}\right).\end{multline}
For $|p^{1/2}|<|x|<|p^{-1/2}|$, 
\begin{align*}
 \frac{x^{-1}\theta(x^2,px^2,-px^2;p^2)}{\theta(p^3x^6;p^6)_\infty}&=C_2\sum_{k=-\infty}^\infty\left(\frac{k+1}3\right)
\frac{p^{k-1}}{1+p^{2k-1}}\,x^{2k-1}.\\
\frac{x^{-2}\theta(x^2,-x^2,px^2;p^2)}{\theta(p^3x^6;p^6)_\infty}&=C_4\sum_{k=-\infty}^\infty\left(\frac{k}3\right)
\frac{p^{k-1}}{1+p^{2k}}\,x^{2k}.
\end{align*}
Finally, for $|p^{1/3}|<|x|<|p^{-1/3}|$,
$$ \frac{\theta(x^2,-x^2,px^2;p^2)}{x\theta(-x^2;p^{2/3})_\infty}=C_3\sum_{k=-\infty}^\infty
(-1)^{k}\frac{1-p^{(2k-1)/3}}{1-p^{2k-1}}\,p^{(2k-2)/3}x^{2k-1}. $$
\end{lemma}

Tor prove Lemma \ref{ll} we
 will need the following identities.
As they imply \eqref{csf},
one can view Lemma \ref{qtl} as a more explicit version of Lemma~\ref{tql}.

\begin{lemma}
\label{qtl}
The following identities hold:
\begin{multline*}x^{-1}\theta(x^2,px^2,-px^2;p^2)=\frac{(p^6;p^6)_\infty}{(p^2;p^2)_\infty}
\left(x^{-1}\theta(-p^2x^6;p^6)-x\theta(-p^4x^6;p^6)\right),\\
=\frac{(p^{2/3};p^{2/3})_\infty}{(1-\omega)(p^2;p^2)_\infty}
\,x^{-1}\left(\theta(-\omega x^2;p^{2/3})-\omega\theta(-\omega^2x^2;p^{2/3})\right),\end{multline*}
\begin{multline*} x^{-2}\theta(x^2,-x^2,px^2;p^2)=\frac{(p^6;p^6)_\infty}{(p^2;p^2)_\infty}
\left(x^{-2}\theta(-px^6;p^6)-x^2\theta(-p^5x^6;p^6)\right),\\
=\frac{p^{-1/3}(p^{2/3};p^{2/3})_\infty}{(\omega-\omega^2)(p^2;p^2)_\infty}
\left(\theta(-\omega^2p^{1/3} x^2;p^{2/3})-\theta(-\omega p^{1/3}x^2;p^{2/3})\right).
\end{multline*}
\end{lemma}

\begin{proof}
As in the proof of Lemma \ref{tql}, we start from the quintuple product identity in the form
$$x^{-1}\theta(x^2,px^2,-px^2;p^2)=\frac 1{(p^2;p^2)_\infty}\sum_{n=-\infty}^\infty\left(\frac{n+1}3\right)p^{\frac{n(n-1)}3}x^{2n-1}. $$
Splitting the sum into  the terms corresponding to $n\equiv 0$ and $n\equiv 1\ \operatorname{mod}\ 3$, each group of terms is summed by the triple product identity \eqref{jtp}. This gives the first identity.
Writing instead
$$\left(\frac{n+1}3\right)=\frac{\omega^n-\omega^{1+2n}}{1-\omega},$$
we obtain in the same way the second identity. Replacing $x$ by $\sqrt px$ 
 gives 
the remaining results.
\end{proof}

\begin{proof}[Proof of \emph{Lemma \ref{ll}}]
We focus on the identity \eqref{sri}. By Lemma \ref{qtl} and \eqref{kr},
\begin{align*}
\frac{x\theta(x^2,-x^2,px^2;p^2)}{\theta(-x^6;p^6)_\infty}
&=\frac{(p^6;p^6)_\infty}{(p^2;p^2)_\infty}\left(\frac{x\theta(-px^6;p^6)}{\theta(-x^6;p^6)}-\frac{x^5\theta(-p^5x^6;p^6)}{\theta(-x^6;p^6)}\right)\\
&=C\left(\sum_{n=-\infty}^\infty\frac{(-1)^nx^{6n+1}}{1-p^{6n+1}}-\sum_{n=-\infty}^\infty\frac{(-1)^nx^{6n+5}}{1-p^{6n+5}}\right)\\
&=C\sum_{k=-\infty}^\infty\left(\frac{k+1}3\right)
\frac{(-1)^{k}x^{2k-1}}{1-p^{2k-1}},\end{align*}
where $|p|<|x|<1$
and
\begin{align*}C&=\frac{\theta(p;p^6)}{(p^2;p^2)_\infty(p^6;p^6)_\infty}=\frac{(p,p^5;p^6)}{(p^2;p^2)_\infty(p^6;p^6)_\infty}=\frac{(p;p^2)_\infty}{(p^2;p^2)_\infty(p^3;p^3)_\infty}\\
&=\frac{(p;p)_\infty}{(p^2;p^2)_\infty^2(p^3;p^3)_\infty}=C_1. \end{align*}
 The case $p=0$ is 
$$\frac{x^{-2}-x^2}{x^{-3}+x^3}=\sum_{k=1}^\infty\left(\frac{k+1}3\right)
(-1)^{k}x^{2k-1},\qquad |x|<1. $$
Combining these two summations leads to the desired identity in $|p|<|x|<1$. 
As the resulting series is convergent  for $|p|< |x|<|p^{-1}|$, it holds in the larger annulus by analytic continuation.

The remaining expansions follow even more easily from Lemma \ref{qtl} and \eqref{kr}.
\end{proof}

\subsection{Schur polynomial expansions}

We will now derive expansions of the 
functions $P_n^{(\sigma)}$, with $\sigma \neq 1,\hat 1$, into Schur polynomials; 
the other two cases are considered in \S \ref{sqs}.
We will write
$$\chi_{\mu}(x)=\chi_{\mu_1,\dots,\mu_m}(x_1,\dots,x_m)=\det_{1\leq i,j\leq m}(x_i^{\mu_j}),$$
where $\mu_j$ are distinct integers. After reordering the columns
and factoring out a power of $X=x_1\dotsm x_m$, one may assume that
$\mu_j=\lambda_j+m-j$, where $\lambda_1\geq\dots\geq \lambda_m\geq 0$
is a partition. Then,
$$\chi_\mu(x)=\Delta(x)s_\lambda(x) $$
where $s_\lambda$ is a Schur polynomial and $\Delta$ is given in \eqref{vdn}.

Consider in general a Laurent series
$$\phi(x)=\sum_{k=-\infty}^\infty c_k x^k, $$
convergent in an annulus containing $|x|=1$. Assuming that $\phi(1/x)=-\phi(x)$, or equivalently $c_{-k}=-c_k$,
we may consider the pfaffian
\begin{multline*}
\pfaff_{1\leq i,j\leq 2n}(\phi(x_i/x_j))\\
\begin{split}&=\frac 1{2^n n!}\sum_{\sigma\in S_{2n}}\sgn(\sigma)\sum_{k_1,\dots,k_n\in\mathbb  Z}
c_{k_1}\dotsm c_{k_n}\left(\frac{x_{\sigma(1)}}{x_{\sigma(2)}}\right)^{k_1}\dotsm \left(\frac{x_{\sigma(2n-1)}}{x_{\sigma(2n)}}\right)^{k_n}\\
&=\frac 1{2^n n!} \sum_{k_1,\dots,k_n\in\mathbb  Z}c_{k_1}\dotsm c_{k_n}\chi_{k_1,-k_1,\dots,k_n,-k_n}(x_1,\dots,x_{2n})\\
&=\sum_{1\leq k_1<k_2<\dots<k_n}c_{k_1}\dotsm c_{k_n}\chi_{k_n,\dots,k_1,-k_1,\dots,-k_n}(x_1,\dots,x_{2n}),
\end{split}\end{multline*}
where we used in the last step that the summand is symmetric under permutations of the summation variables as well as reflections $k_j\mapsto -k_j$.
In particular, if $\phi$  is even we may write $\phi(x)=\sum_{k=-\infty}^\infty c_k x^{2k}$ and
$$
\pfaff_{1\leq i,j\leq 2n}(\phi(x_i/x_j))=
\sum_{1\leq k_1<k_2<\dots<k_n}c_{k_1}\dotsm c_{k_n}\chi_{k_n,\dots,k_1,-k_1,\dots,-k_n}(x_1^2,\dots,x_{2n}^2).$$
Similarly, if  $\phi(x)=\sum_{k=-\infty}^\infty c_k x^{2k-1}$ we get
$$
\pfaff_{1\leq i,j\leq 2n}(\phi(x_i/x_j))=X\sum_{1\leq k_1<k_2<\dots<k_n}{c_{k_1}\dotsm c_{k_n}}\chi_{k_n-1,\dots,k_1-1,-k_1,\dots,-k_n}(x_1^2,\dots,x_{2n}^2).$$
Specializing to the Laurent series of Lemma \ref{ll} gives the following 
 expansions of the functions $P_n^{(\sigma)}$, where $\sigma=0,2,3,4,6$. 
Corresponding expansions for $P_n^{\hat\sigma}$ follow using \eqref{phs}.

\begin{proposition}\label{spp}
In a neighbourhood of $|x_1|=\dots=|x_{2n}|=1$, we have the expansions
\begin{align*}
P_n^{(0)}&=C_0^n\prod_{1\leq i<j\leq 2n}\theta(p^{1/3}x_i^2/x_j^2;p^{2/3})\sum_{1\leq k_1<\dots<k_n}\prod_{j=1}^n\frac{1-p^{4k_j/3}}{1+p^{2k_j}}\,p^{(k_j-1)/3}\\
&\quad\times\chi_{k_n,\dots,k_1,-k_1,\dots,-k_n}(x_1^2,\dots,x_{2n}^2),\\
P_n^{(2)}&={C_2^n}{X}\prod_{1\leq i<j\leq 2n}\theta(p^{3}x_i^6/x_j^6;p^{6})\sum_{1\leq k_1<\dots<k_n}\prod_{j=1}^n\left(\frac{k_j+1}3\right)\frac{p^{k_j-1}}{1+p^{2k_j-1}}\\
&\quad\times\chi_{k_n-1,\dots,k_1-1,-k_1,\dots,-k_n}(x_1^2,\dots,x_{2n}^2),\\
P_n^{(3)}&={C_3^n}X\prod_{1\leq i<j\leq 2n}\frac{x_j}{x_i}\theta(-x_i^2/x_j^2;p^{2/3})\sum_{1\leq k_1<\dots<k_n}\prod_{j=1}^n(-1)^{k_j}\frac{1-p^{(2k_j-1)/3}}{1-p^{2k_j-1}}\,p^{(2k_j-2)/3}\\
&\quad\times\chi_{k_n-1,\dots,k_1-1,-k_1,\dots,-k_n}(x_1^2,\dots,x_{2n}^2),\\
P_n^{(4)}&=C_4^n\prod_{1\leq i<j\leq 2n}\theta(p^{3}x_i^6/x_j^6;p^{6})\sum_{1\leq k_1<\dots<k_n}\prod_{j=1}^n\left(\frac{k_j}3\right)\frac{p^{k_j-1}}{1+p^{2k_j}}\\
&\quad\times\chi_{k_n,\dots,k_1,-k_1,\dots,-k_n}(x_1^2,\dots,x_{2n}^2),\\
P_n^{(6)}&={C_6^n}X\prod_{1\leq i<j\leq 2n}\theta(p^{1/3}x_i^2/x_j^2;p^{2/3})\sum_{1\leq k_1<\dots<k_n}\prod_{j=1}^n\frac{1-p^{(4k_j-2)/3}}{1+p^{2k_j-1}}\,p^{(k_j-1)/3}\\
&\quad\times\chi_{k_n-1,\dots,k_1-1,-k_1,\dots,-k_n}(x_1^2,\dots,x_{2n}^2),
\end{align*}
where the constants $C_\sigma$ are given in \eqref{ci}.
\end{proposition}

\subsection{Symmetric function expansions of $P_n^{(1)}$}\label{sqs}

To obtain an analogue of Proposition \ref{spp} for $P_n^{(1)}$, the
 following  result is useful.

\begin{lemma}\label{pel}
Let $\Lambda$ be  a countable and totally ordered set,
let $A$ be a skew-symmetric $(2n\times 2n)$-matrix and let $B^{(k)}$, $C^{(k)}$, $k\in\Lambda$, be $2n$-dimensional
column vectors.
  For $k_1,\dots,k_m\in\Lambda$, let $X_{k_1,\dots,k_m}$ be the $(2n\times 2m)$-matrix
$$X_{k_1,\dots,k_m}=\left(B^{(k_m)}\dotsm B^{(k_1)} C^{(k_1)}\dotsm C^{(k_m)}\right). $$
Then, 
\begin{multline}\label{pli}
\pfaff_{1\leq i,j\leq 2n}\left(A_{ij}+\sum_{k\in\Lambda}\left(B_i^{(k)}C_j^{(k)}-B_j^{(k)}C_i^{(k)}\right)\right)\\
=\sum_{m=0}^n \,\sum_{\substack{k_1,\dots,k_m\in\Lambda\\k_1<\dots<k_m}}\pfaff\left(\begin{matrix} A & X_{k_1,\dots,k_m} \\
-X_{k_1,\dots,k_m}^T & 0\end{matrix}\right),
\end{multline}
where we assume that the sums converge absolutely if $\Lambda$ is infinite.
\end{lemma}

\begin{proof}
Let $P$ denote the left-hand side of \eqref{pli}. By definition,
\begin{align*}
P&=\frac 1{2^n n!}\sum_{\sigma\in S_{2n}}\sgn(\sigma)\prod_{i=1}^n\left(A_{\sigma(2i-1)\sigma(2i)}+\sum_{k\in\Lambda}
\left(B_{\sigma(2i-1)}^{(k)}C_{\sigma(2i)}^{(k)}-B_{\sigma(2i)}^{(k)}C_{\sigma(2i-1)}^{(k)}\right)\right)\\
&=\frac 1{2^n n!}\sum_{S\subseteq[1,n]}\sum_{\sigma\in S_{2n}}\sgn(\sigma)\prod_{i\notin S}A_{\sigma(2i-1)\sigma(2i)}\\
&\hspace*{6cm}\times\prod_{i\in S}\sum_{k\in\Lambda}
\left(B_{\sigma(2i-1)}^{(k)}C_{\sigma(2i)}^{(k)}-B_{\sigma(2i)}^{(k)}C_{\sigma(2i-1)}^{(k)}\right).
\end{align*}
Writing $S=\{s_1<\dots<s_m\}$, $S^c=\{t_1<\dots<t_{n-m}\}$, let $\rho$ be the even permutation
$$\rho(1,\dots,2n)=(2t_1-1,2t_1,\dots,2t_{n-m}-1,2t_{n-m},2s_1-1,2s_1,\dots,2s_m-1,2s_m). $$
 Replacing $\sigma$ by $\sigma\rho^{-1}$ in the sum gives
\begin{align*}
P&=\frac 1{2^n n!}\sum_{m=0}^n\binom nm\sum_{\sigma\in S_{2n}}\sgn(\sigma)\prod_{i=1}^{n-m}A_{\sigma(2i-1)\sigma(2i)}\\
&\quad\times\prod_{i=n-m+1}^n\sum_{k\in\Lambda}
\left(B_{\sigma(2i-1)}^{(k)}C_{\sigma(2i)}^{(k)}-B_{\sigma(2i)}^{(k)}C_{\sigma(2i-1)}^{(k)}\right)\\
&=\frac 1{2^n n!}\sum_{m=0}^n\binom nm\sum_{k_1,\dots,k_m\in\Lambda}\,\sum_{\sigma\in S_{2n}}\sgn(\sigma)\prod_{i=1}^{n-m}A_{\sigma(2i-1)\sigma(2i)}\\
&\quad\times\prod_{i=1}^m
\left(B_{\sigma(2i+2n-2m-1)}^{(k_i)}C_{\sigma(2i+2n-2m)}^{(k_i)}-B_{\sigma(2i+2n-2m)}^{(k_i)}C_{\sigma(2i+2n-2m-1)}^{(k_i)}\right).
\end{align*}
We observe that the summand is symmetric in the indices $k_j$.
Moreover, expanding the differences lead to terms that can be identified by a change of $\sigma$;  the minus signs are
then incorporated in the factor $\sgn(\sigma)$. It follows that
$$ P=\sum_{m=0}^n\sum_{k_1<\dots<k_m}P_{k_1,\dots,k_m},$$
where
\begin{align*}
P_{k_1,\dots,k_m}&=\frac 1{2^{n-m}(n-m)!}\sum_{\sigma\in S_{2n}}\sgn(\sigma)\prod_{i=1}^{n-m}A_{\sigma(2i-1)\sigma(2i)}\\
&\quad\times\prod_{i=1}^m B_{\sigma(2i+2n-2m-1)}^{(k_i)}C_{\sigma(2i+2n-2m)}^{(k_i)}.\end{align*}

For each $\sigma\in S_{2n}$, consider the pairing on $[1,2n+2m]$ defined by
$\sigma(2i-1)\sim\sigma(2i)$ for $1\leq i\leq n-m$ and $\sigma(i)\sim i+2m$ for
$2n-2m+1\leq i\leq 2n$. 
We observe  that no two elements in $[2n+1,2n+2m]$
are paired together. Moreover, there are $2^{n-m}(n-m)!$ permutations corresponding to each such
pairing. This allows us to identify $P_{k_1,\dots,k_m}$ with
the pfaffian on the right-hand side of \eqref{pli},
but with $X_{k_1,\dots,k_m}$ replaced by
$$\left(B^{(k_1)}  C^{(k_1)}\dotsm B^{(k_m)} C^{(k_m)}\right).$$
As  this differs from $X_{k_1,\dots,k_m}$ by an even permutation of the columns, \eqref{pli} holds.
\end{proof}

For $m+n$ even and $\lambda\in\mathbb Z^m$, we introduce the symmetric Laurent polynomials
$$T_{\lambda}(x_1,\dots,x_{n})=\prod_{1\leq i<j\leq n}\frac{x_j^3+x_i^3}{x_j^2-x_i^2}\,\pfaff\left(\begin{matrix}A & B\\ -B^T & 0\end{matrix}\right), $$
where $A$ is the $n\times n$-matrix with matrix elements $(x_j^2-x_i^2)/(x_j^3+x_i^3)$ and $B$ the $n\times m$ matrix with elements
$x_i^{\lambda_j}$. Formally, these are very similar to Schur $Q$-polynomials (or $t=-1$ Hall-Littlewood polynomials), which 
are given by a similar identity, where $\lambda$ is a strict partition and $(x_j^2-x_i^2)/(x_j^3+x_i^3)$ is replaced by $(x_j-x_i)/(x_j+x_i)$ \cite{n}.
Substituting the expansion \eqref{sri} into the definition of $P_n^{(1)}$ and using Lemma \ref{pel} gives the following result.

\begin{proposition}\label{sqe}
Assuming $|p|<|x_i/x_j|<p^{-1}$ for $1\leq i,j\leq 2n$, we have
\begin{align*}
P_n^{(1)}&=(-1)^nC_1^n X^{4-6n}\Delta(x^4)\prod_{1\leq i<j\leq 2n}(-p^6x_i^6/x_j^6,-p^6x_j^6/x_i^6;p^6)_\infty\\
&\quad\times\sum_{m=0}^n\,\sum_{1\leq k_1<\dots<k_m}\prod_{i=1}^m\left(\frac{k_i+1}3\right)\frac{(-1)^{k_i-1}p^{2k_i-1}}{1-p^{2k_i-1}}\\
&\quad\times T_{k_m-1,\dots,k_1-1,-k_1,\dots,-k_m}(x_1^2,\dots,x_{2n}^2).
\end{align*}

\end{proposition}

\subsection{Trigonometric limit}

We will now consider the behaviour of the functions $P_n^{(\sigma)}$  in the trigonometric limit
 $p\rightarrow 0$. 
By the discussion at the end of \S \ref{mts}, this is equivalent to
considering the behaviour of $P_n^{(1)}$ (say) for $\tau$ near a cusp of $\Gamma_0(2,6)$,
that is, as $p$  tends to $0$ or to any twelfth  root of unity.

Recall that the elementary symmetric functions are defined by
$$e_n(x_1,\dots,x_m)=\sum_{1\leq k_1<\dots<k_n\leq m}x_{k_1}\dotsm x_{k_n}. $$

\begin{proposition}\label{tp}
Let $\pi_n^{(\sigma)}$ be the leading coefficient in the Taylor expansion of $P_n^{(\sigma)}$
at $p=0$. Then,
\begin{align*}\pi_n^{(0)}&=(-1)^n p^{\frac{n(n-1)}6}X^{-2n}\Delta(x^2)\,e_n(x_1^2,\dots,x_{2n}^2), \\
\pi_n^{(1)}&=(-1)^nX^{4-6n}\Delta(x^4)\,s_{n-1,n-1,\dots,1,1,0,0}(x_1^4,\dots,x_{2n}^4),\\
\pi_n^{(2)}&=\begin{cases}
(-1)^{\frac{n+1}2}p^{\frac{(n-1)(3n+1)}4}X^{2-3n}\Delta(x^2)s_{n-1,n-1,\dots,1,1,0,0}(x_1^2,\dots,x_{2n}^2), & n \text{ odd},\\
(-1)^{\frac n2}p^{\frac{n(3n-2)}4}X^{1-3n}\Delta(x^2)s_{n,n-1,n-1,\dots,1,1,0}(x_1^2,\dots,x_{2n}^2), & n \text{ even},\end{cases}\\
\pi_n^{(3)}&=(-1)^{\frac{n(n+1)}2}p^{\frac{n(n-1)}3}X^{2-4n}\Delta(x^4),\\
\pi_n^{(4)}&=\begin{cases}(-1)^{\frac{n+1}2}p^{\frac{(n-1)(3n-1)}4}X^{1-3n}\Delta(x^2)s_{n,n-1,n-1,\dots,1,1,0}(x_1^2,\dots,x_{2n}^2), & n \text{ odd},\\
(-1)^{\frac n2}p^{\frac{n(3n-4)}4}X^{2-3n}\Delta(x^2)s_{n-1,n-1,\dots,1,1,0,0}(x_1^2,\dots,x_{2n}^2), & n \text{ even},
\end{cases}\\
\pi_n^{(6)}&=(-1)^np^{\frac{n(n-1)}6}X^{1-2n}\Delta(x^2).
 \end{align*}
\end{proposition}

\begin{proof}
The case $\sigma=1$ was obtained in \eqref{tsp}.
For the other cases, we use Proposition \ref{spp}. In the cases  $\sigma=0,\,3,\,6$, 
only the term with $k_j=j$ for all $j$ contributes to
the leading coefficient. We find that
\begin{align*}
\pi_n^{(0)}&=(-1)^n\prod_{j=1}^n p^{(j-1)/3}\chi_{n,n-1,\dots,1,-1,-2,\dots,-n}(x_1^2,\dots,x_{2n}^2)\\
&=(-1)^np^{\frac{n(n-1)}6}X^{-2n}\Delta(x^2)\,s_{1^n0^n}(x_1^2,\dots,x_{2n}^2).
\end{align*}
For the fact that $s_{1^n0^n}=e_n$, see e.g. \cite[Eq.\ (I.3.9)]{ma}.
For $\sigma=3$ and $\sigma=6$ the argument is similar but we encounter instead the Schur polynomial
$s_{0^{2n}}=1$. 

For $\sigma=2$, we obtain the leading coefficient when $k_j$ is the $j$-th smallest positive integer not congruent to $2$ mod $3$, that is,
$k_j=[3j/2]$. It is easy to verify that 
$$\prod_{j=1}^{n}\left(\frac{k_j+1}3\right)=(-1)^{[(n+1)/2]}, $$
$$\prod_{j=1}^n p^{k_j-1}=\begin{cases}
p^{(n-1)(3n+1)/4}, & n \text{ odd},\\
p^{n(3n-2)/4}, & n \text{ even},
\end{cases}
$$
\begin{multline*}\chi_{k_n-1,\dots,k_1-1,-k_1,\dots,-k_n}(x_1^2,\dots,x_{2n}^2)\\
=\begin{cases}
X^{1-3n}\Delta(x^2)s_{n-1,n-1,\dots,1,1,0,0}(x_1^2,\dots,x_{2n}^2), & n \text{ odd},\\
X^{-3n}\Delta(x^2)s_{n,n-1,n-1,\dots,1,1,0}(x_1^2,\dots,x_{2n}^2), & n \text{ even}.\end{cases}
\end{multline*}
The final case $\sigma=4$ is similar.

\end{proof}

\subsection{Hankel determinants}

In the applications to statistical mechanics, one is  
particularly interested in the homogeneous limit, when all the variables  $x_j$ coincide.
By the following fact, the homogeneous limit of the pfaffians $P_n^{(\sigma)}$ can be expressed in terms
of Hankel determinants. Although the result is presumably well-known, we include a proof for completeness.

\begin{lemma}\label{phl}
For $f$ a sufficiently differentiable odd function,
$$\lim_{z_1,\dots,z_{2n}\rightarrow 0}\frac{\pfaff_{1\leq i,j\leq 2n}(f(z_i-z_j))}{\prod_{1\leq i<j\leq 2n}(z_i-z_j)}=\frac{\det_{1\leq i,j\leq n}(f^{(2i+2j-3)}(0))}{\prod_{j=1}^{2n}(j-1)!}. $$
\end{lemma}

\begin{proof}
Let 
$$F(z_1,\dots,z_{2n})=\pfaff_{1\leq i,j\leq 2n}(f(z_i-z_j))=\frac{1}{2^n n!}\sum_{\sigma\in S_{2n}}\sgn(\sigma)\prod_{j=1}^n f(z_{\sigma(2j-1)}-z_{\sigma(2j)}).$$
 As $F$ is anti-symmetric, the lowest term in the Taylor series for $F$ has the form
$C\Delta(z)$, where 
\begin{align*}C&=(-1)^n\prod_{j=1}^{2n}\frac 1{(j-1)!}\frac{\partial ^{j-1}}{\partial z_j^{j-1}}\Bigg|_{z_j=0} F(z_1,\dots,z_{2n})\\
&=\frac{(-1)^n}{2^n n!}\prod_{j=1}^{2n}\frac 1{(j-1)!}\sum_{\sigma\in S_{2n}}\sgn(\sigma)\prod_{j=1}^n(-1)^{\sigma(2j)-1} f^{(\sigma(2j-1)+\sigma(2j)-2)}(0).
\end{align*}
Since $f$ is odd, only permutations such that $\sigma(2j-1)\not\equiv \sigma(2j)\ \operatorname{mod} 2$ for each $j$ contribute to the sum. For any such permutation, define $\tau\in S_n$ by 
$$\{\sigma(2j-1),\sigma(2j)\}=\{2k-1,2\tau(k)\}.$$
 Then, each $\tau$ corresponds to $2^n n!$ 
choices of $\sigma$. Moreover, it is easy to check that $\sgn(\sigma)\prod_{j=1}^n(-1)^{\sigma(2j)}=\sgn(\tau)$. It follows that
$$C=\prod_{j=1}^{2n}\frac 1{(j-1)!}\sum_{\tau\in S_{n}}\sgn(\tau)\prod_{k=1}^n f^{(2k+2\tau(k)-3)}(0)=\frac{\det_{1\leq i,j\leq n}(f^{(2i+2j-3)}(0))}{\prod_{j=1}^{2n}(j-1)!}. $$
\end{proof}

To apply Lemma \ref{phl} to the pfaffians $P_n^{(\sigma)}$, it is convenient to rewrite the series in Lemma \ref{ll} in trigonometric form.
For instance, for $\sigma=3$ we have as all the variables $z_j\rightarrow 0$
$$P_n^{(3)}\sim (2(-p^{2/3};p^{2/3})_\infty)^{n(2n-1)}\pfaff_{1\leq i,j\leq 2n}(f(z_i-z_j)), $$ 
where
\begin{align*} 
f(z)&=C_3\sum_{k=-\infty}^\infty (-1)^k\frac{1-p^{(2k-1)/3}}{1-p^{2k-1}}\,p^{(2k-2)/3}\,e^{(2k-1)\ti\pi z}\\
&=2\ti C_3\sum_{k=1}^\infty (-1)^k\frac{1-p^{(2k-1)/3}}{1-p^{2k-1}}\,p^{(2k-2)/3}\,\sin\big((2k-1)\pi z\big).
\end{align*}
Applying Lemma \ref{phl} gives
\begin{multline*}\lim_{z_1,\dots,z_n\rightarrow 0}\frac{P_n^{(3)}(z)}{\Delta(z)}=\frac{(2(-p^{2/3};p^{2/3})_\infty)^{n(2n-1)}(2\ti C_3)^n}{\prod_{j=1}^n(j-1)!}\\
\det_{1\leq i,j\leq n}\left((-1)^{i+j}\pi^{2i+2j-3}\sum_{k=1}^\infty(-1)^k\frac{1-p^{(2k-1)/3}}{1-p^{2k-1}}\,p^{(2k-2)/3}(2k-1)^{2i+2j-3}\right).
 \end{multline*}
By linearity, the factor $(-1)^{i+j}\pi^{2i+2j-2}$ inside the determinant can be replaced by a global factor $\pi^{n(2n-1)}$.
This gives the case $\sigma=3$ of Theorem \ref{hdt} below.

In the case $\sigma=1$, the term $(x^{-2}-x^2)/(x^{-3}+x^3)$ in \eqref{sri}  contributes a Taylor coefficient of
$\sin(2z)/\cos(3z)$ to each determinant entry. We write these in terms of Glaisher's $T$-numbers \cite{s}
$$1,\ 23,\ 1681,\ 257543,\ 67637281,\ \dots, $$
which are given by the generating function
$$\frac{\sin(2z)}{2\cos(3z)}=\sum_{n=0}^\infty\frac{T_n z^{2n+1}}{(2n+1)!}. $$
We obtain in this way the following Hankel determinant formulas.

\begin{theorem}\label{hdt}
 For $\sigma=0,1,2,3,4,6$, we have
$$\lim_{z_1,\dots,z_{2n}\rightarrow 0}\frac{P_n^{(\sigma)}(z_1,\dots,z_{2n})}{\prod_{1\leq i<j\leq 2n}(z_i-z_j)}
= \frac{(2\ti C_\sigma)^nD_\sigma^{n(2n-1)}}{\prod_{j=1}^{2n}(j-1)!}\,H_n^{(\sigma)},$$
with  $C_\sigma$ as in \eqref{ci},
\begin{align*}
D_0&=D_6=\pi(p^{1/3};p^{2/3})_\infty^2,&
D_1&=2\pi(-p^6;p^6)_\infty^2,\\
D_2 &=D_4=\pi(p^3;p^6)_\infty^2,&
D_3&=2\pi(-p^{2/3};p^{2/3})_\infty^2
\end{align*}
and
\begin{equation}\label{hd}H_n^{(\sigma)}=\det_{1\leq i,j\leq n}\left(L_{i+j-2}^{(\sigma)}\right),\end{equation}
where $L_j^{(\sigma)}$ is the Lambert series
\begin{align*}
L_j^{(0)}&=\sum_{k=1}^\infty\frac{1-p^{4k/3}}{1+p^{2k}}\,p^{(k-1)/3}(2k)^{2j+1},\\
L_j^{(1)}&=(-1)^{j+1}T_j+\sum_{k=1}^\infty\left(\frac{k+1}3\right)\frac{(-1)^{k}p^{2k-1}}{1-p^{2k-1}}(2k-1)^{2j+1},\\
L_j^{(2)}&=\sum_{k=1}^\infty\left(\frac{k+1}3\right)\frac{p^{k-1}}{1+p^{2k-1}}(2k-1)^{2j+1},\\
L_j^{(3)}&=\sum_{k=1}^\infty(-1)^k\frac{1-p^{(2k-1)/3}}{1-p^{2k-1}}\,p^{(2k-2)/3}(2k-1)^{2j+1},\\
L_j^{(4)}&=\sum_{k=1}^\infty\left(\frac{k}3\right)\frac{p^{k-1}}{1+p^{2k}}(2k)^{2j+1},\\
L_j^{(6)}&=\sum_{k=1}^\infty\frac{1-p^{(4k-2)/3}}{1+p^{2k-1}}\,p^{(k-1)/3}(2k-1)^{2j+1}.
\end{align*}
\end{theorem}

We hope that this result will be useful for studying the subsequent limit
 $n\rightarrow\infty$, which is of interest in statistical mechanics.
One approach would be to identify the matrix elements 
with moments and then use the corresponding orthogonal polynomials (see e.g.\ \cite{bh,cp} for applications to the six-vertex model).
As a step in this direction,
we have observed that the integral evaluation
$$\frac{\sin(2z)}{2\cos(3z)}=\frac 1{4\sqrt 3}\int_{-\infty}^\infty\frac{\sinh(\pi y/3)}{\cosh(\pi y/2)}\,e^{yz}\,dy,\qquad |\operatorname{Re} (z)|<\frac{\pi}6 $$
leads to the moment representation
\begin{align*}T_n&=\frac{1}{2\sqrt 3}\int_{0}^\infty\frac{\sinh(\pi y/3)}{\cosh(\pi y/2)}\,y^{2n+1}\,dy\\
&=\frac {3\sqrt 3\cdot 6^{2n}}{4\pi^2}\int_{0}^\infty\left|\frac{\Gamma(1/6+\ti x)\Gamma(1/2+\ti x)\Gamma(5/6+\ti x)}{\Gamma(2\ti x)}\right|^2x^{2n}\,dx, \end{align*}
where we substituted $y=6x$ and used some well-known identities for the gamma function.
The corresponding orthogonal polynomials are the continuous dual Hahn polynomials 
$S_n(x^2;1/6,1/2,5/6)$, see \cite{ks}. Thus,  the Hankel determinant $H_n^{(1)}$
relates to a perturbation of these polynomials.

\section{Applications in statistical mechanics}
\label{as}

\subsection{Domain wall partition function for the 8VSOS model}

We will consider the inhomogeneous 
8VSOS model with domain wall boundary conditions. 
We follow the conventions of \cite{r}.
A state of the model is an assignment
 of a height function to  the squares  of an $n\times n$ 
chessboard,  such that the heights of neighbouring squares differ by exactly $1$. Moreover, the heights of the boundary squares are fixed by demanding that
the north-west and south-east squares have height $0$ and that, moving away from these squares along the boundary, the height increases. Consequently, the north-east and south-west squares have height $n$. 
We assign local weights to $2\times 2$-blocks of adjacent squares. These blocks are given matrix coordinates $1\leq i,j\leq n$ in a standard way.
A block with coordinates $(i,j)$ and heights $\left[\begin{smallmatrix}a&b\\c&d\end{smallmatrix}\right]$
  has weight $R^{b-a,d-b}_{d-c,c-a}(\lambda q^a,u_i/v_j)$, where
\begin {align*}R^{++}_{++}(\lambda,u)&=R^{--}_{--}(\lambda,u)=\frac{\theta(q u;p)}{\theta(q;p)}, \\
R^{+-}_{+-}(\lambda,u)&=\frac{\theta(u,q\lambda;p)}{\theta(q,\lambda;p)},\qquad R^{-+}_{-+}(\lambda,u)=q\frac{\theta(u,q^{-1}\lambda ;p)}{\theta(q,\lambda;p)},\\
R^{-+}_{+-}(\lambda,u)&=\frac{\theta(\lambda u;p)}{\theta(\lambda;p)},\qquad R^{+-}_{-+}(\lambda,u)=u\frac{\theta(\lambda/u;p)}{\theta(\lambda;p)}.\end{align*}
The weight of  a state is the product of the weights of all $2\times 2$-blocks. Finally, the partition function
 $$Z_n=Z_n(u_1,\dots,u_n;v_1,\dots,v_n;\lambda;p,q)$$ 
is the sum of the weights of all the states.
As a manifestation of the integrability of the model,
$Z_n$ is separately symmetric in the parameters
$u_1,\dots,u_n$ and $v_1,\dots,v_n$.

If we first let $p=0$ and then $\lambda=0$, 
 the 8VSOS model reduces to the six-vertex model. The domain wall partition function is then given by
the Izergin--Korepin determinant \eqref{ik}. More precisely,
\begin{multline*}Z_n(u_1,\dots,u_n;v_1,\dots,v_n;0;0,q)\\
=\frac{(-1)^{\binom n2}U}{(1-q)^{2\binom n2}V^n}\frac{\prod_{i,j=1}^n(u_i-v_j)(qu_i-v_j)}{\Delta(u)\Delta(v)}\,\det_{1\leq i,j\leq n}\left(\frac 1{(u_i-v_j)(qu_i-v_j)}\right)\end{multline*}
(as before, we write $U=u_1\dotsm u_n$ etc.). In particular, if in addition $q=\omega=e^{2\ti\pi/3}$, \eqref{si} gives
\begin{equation}\label{npfb}3^{\binom n2}\omega^{n(n-2)}\frac{V^n}{U}\,Z_n(\omega u;v;0;0,\omega)=s_{n-1,n-1,\dots,1,1,0,0}(u,v). \end{equation}

For the general 8VSOS model, explicit expressions for the domain wall partition function are discussed in \cite{g1,g2,p,r0,rf}.
We will now give a new expression in the case $q=\omega$, by showing that the  partition function is a linear combination
of the pfaffians $P_n^{(2)}$ and $P_n^{(4)}$.

\begin{theorem}\label{dwt}
Consider the domain wall partition function for the \emph{8VSOS} model, with  $p=e^{\ti\pi\tau}$ and $q=\omega=e^{2\ti\pi/3}$.  Introduce the parameters $z_j$ and $x_j$ by $e^{\ti\pi z_j}=x_j$ and
$$(e^{2\ti\pi z_1},\dots,e^{2\ti\pi z_n})=(x_1^2,\dots,x_{2n}^2)=( u_1,\dots,u_n,v_1,\dots,v_n). $$ 
Writing
$$\Delta(x^2;p)=\prod_{1\leq i<j\leq 2n}x_i^2\theta(x_j^2/x_i^2;p), $$
we then have
\begin{multline}\label{zs}
p^{\binom n2}\theta(p;p^2)\theta(\omega;p)^{n(n-1)}\theta(\lambda\om^{n+1},\lambda\omega^{n+2};p)\Delta(x^2;p)\,Z_n(\omega u;v;\lambda;p;\omega)\\
=(-1)^n{\omega^{2n}\lambda^{n+1}X^{2n-1}\theta(-p^2;p^6)^{n-1}\theta(-p^{n+1}\omega^{n}\lambda^2 V/U;p^2)}\,P_n^{(2)}(z_1,\dots,z_{2n}|\tau)\\
+\lambda^{n}UX^{2n-2}\theta(-p;p^6)^{n-1}\theta(-p^{n}\omega^{n}\lambda^2 V/U;p^2)\,P_n^{(4)}(z_1,\dots,z_{2n}|\tau).
\end{multline}
\end{theorem}

The splitting of the partition function as a sum of two parts is central for our previous investigation of the three-colour model
\cite{r0,r}.  What is new here is not only the explicit expression for each term as a pfaffian, but also the fact that the partition function
is ``almost symmetric'' in all $2n$ spectral parameters, in the sense that the only asymmetry comes from the factors
$\theta(-p^{n+1}\om^2\lambda^2V/U;p^2)$ and $U\theta(-p^{n}\om^2\lambda^2V/U;p^2)$.

\begin{proof} By the discussion at the end of \cite[\S 4]{r}, 
the left-hand side of \eqref{zs} 
as a function of $\lambda$ belongs to the space 
spanned by $\theta(-\om^{n}\lambda^2V/U;p^2)$ and
$\lambda\theta(-p\om^{n}\lambda^2V/U;p^2)$. This is a consequence of the explicit formula for $Z_n$ 
found in \cite{r0}. 
By \eqref{tqp}, we may  as well work with the
basis   $\lambda^n\theta(-p^n\om^{n}\lambda^2V/U)$ and
$\lambda^{n+1}\theta(-p^{n+1}\om^{n}\lambda^2V/U)$.
Thus, we can write 
\begin{multline}\label{zps}
p^{\binom n2}\theta(p;p^2)\theta(\omega;p)^{n(n-1)}\theta(\lambda\om^{n+1},\lambda\omega^{n+2};p)\Delta(x^2;p)\,Z_n(\omega u;v;\lambda;p;\omega)\\
=(-1)^n{\omega^{2n}\lambda^{n+1}X^{2n-1}\theta(-p^2;p^6)^{n-1}\theta(-p^{n+1}\omega^{n}\lambda^2 V/U;p^2)}\,Q_n\\
+{\lambda^{n}UX^{2n-2}\theta(-p;p^6)^{n-1}\theta(-p^{n}\omega^{n}\lambda^2 V/U;p^2)}\,R_n,
\end{multline}
where $Q_n$ and $R_n$ are independent of $\lambda$.  We must prove that $Q_n=P_n^{(2)}$ and $R_n=P_n^{(4)}$.

We will show that $Q_n$ satisfies all conditions of Proposition \ref{ap}, with $\sigma=2$,
when viewed as a function of $z_1$. First of all, the left-hand side of \eqref{zps} is an entire function of $z_1$. 
Adding $(-1)^{n+1}$ times the same function with $\lambda$ replaced by $-\lambda$,
it follows that the first term on the right is entire.
 As $Q_n$ is independent of $\lambda$, we conclude that $Q_n$ is entire.

The identity \eqref{eo} is obvious.
Consider now the effect of replacing
$z_1$ by $z_1+\tau/2$  in \eqref{zps}.
Then, $u_1\mapsto pu_1$. 
It is easy to see that
\begin{equation}\label{zqp}Z_n(p\omega u_1,\dots)=\frac{(-1)^nV\lambda}{\omega^n u_1^n}\,Z_n(\omega u_1,\dots) \end{equation}
(cf.\  \cite[Lemma 3.2]{r0}) and that,
writing \eqref{zps} as $CZ_n=AQ_n+BR_n$, 
\begin{align*}
A\left(z_1+\frac\tau2\right)&=(-1)^n\frac{\lambda p^{n-\frac 12}\theta(-p^2;p^6)^{n-1}X}{\omega^n\theta(-p;p^6)^{n-1}U}\,B(z_1),\\
B\left(z_1+\frac\tau2\right)&=(-1)^n\frac{\lambda p^{2n-1}\theta(-p;p^6)^{n-1}V}{\omega^n\theta(-p^2;p^6)^{n-1}X}\,A(z_1),\\
C\left(z_1+\frac\tau2\right)&=-\frac{X^2}{u_1^{2n}}\,C(z_1).
\end{align*}
This gives
\begin{align}\nonumber Q_n\left(z_1+\frac\tau2\right)&=\frac{B(z_1)C(z_1+\tau/2)Z_n(z_1+\tau/2)}{A(z_1+\tau/2)C(z_1)Z_n(z_1)}\,R_n(z_1)\\\
\label{qsr}&=A_2B_2^{n-1}\frac{x_2^3\dotsm x_{2n}^3}{x_1^{6n-3}}\,R_n(z_1),\end{align}
where $A_2$ and $B_2$ are as in Corollary \ref{hsc}.
Thus, once we have proved that $Q_n=P_n^{(2)}$, the identity  $R_n=P_n^{(4)}$ will follow.
Moreover, we similarly obtain
\begin{equation}\label{rsq}R_n\left(z_1+\frac\tau2\right)
=A_4B_4^{n-1}\frac{x_2^3\dotsm x_{2n}^3}{x_1^{6n-3}}\,Q_n(z_1).\end{equation}
Combining \eqref{qsr} and \eqref{rsq}, we find that $Q_n$ satisfies \eqref{fqp}.

It was proved in \cite{ras0} (see also \cite[Prop.\ 4.4]{r}) that if
$L(u_1,\lambda)$ denotes the left-hand side of \eqref{zps}, then
$$L(u_1,\lambda)+L(\omega u_1,\omega^2\lambda)+L(\omega^2 u_1,\omega\lambda)=0. $$
This implies that $Q_n$ satisfies \eqref{ptq}. 
Finally,  the vanishing conditions \eqref{vc} and \eqref{ve} follow from the corresponding vanishing of 
$\Delta(x^2;p)$.

We have now verified all conditions of Proposition \ref{ap} and conclude that
$Q_n=C_nP_n^{(2)}$, where $C_n$ is independent of $z_1$. We will prove that $C_n=1$ by induction on $n$. 
We start the induction  at $n=0$, with the interpretation $Z_0=P_0^{(2)}=P_0^{(4)}=1$. We must then prove that
\begin{equation}\label{tti}\theta(p;p^2)\theta(\lambda\omega,\lambda\omega^2;p)=\frac{\lambda\theta(-p\lambda^2;p^2)}{\theta(-p^2;p^6)}
+\frac{\theta(-\lambda^2;p^2)}{\theta(-p;p^6)}.\end{equation}
As we have already discussed, the left-hand side is in the space of theta functions spanned by the two terms on the right. 
To compute the coefficients amounts to verifying
 \eqref{tti} for $\lambda=\ti \sqrt p$ and $\lambda =i$. This follows from \eqref{ts}, using also $\theta(x,-x;p)=\theta(x^2;p^2)$.

For the induction step, we 
first note that \cite[Lemma 3.3]{r0}  
\begin{multline*}Z_n(u_1,\dots,u_n;v_1,\dots,v_n;\lambda;p,q)\Big|_{v_1=qu_1}
=\frac{q^{n-2}\theta(\lambda q^n;p)}{\theta(q;p)^{2n-2}\theta(\lambda q^{n-1};p)}\\
\times\prod_{j=2}^n\theta(u_1/v_j,u_j/qu_1;p)\, Z_{n-1}(u_2,\dots,u_n;v_2,\dots,v_n;\lambda;p,q).\end{multline*}
Combining this with \eqref{zqp} gives
\begin{multline*}Z_n(\omega u_1,\dots,\omega u_n;v_1,\dots,v_n;\lambda;p,\omega)\Big|_{u_1=p\omega v_1}=-\frac{\lambda \omega^{n+1}\theta(\lambda \omega^n;p)}{\theta(\omega;p)^{2n-2}\theta(\lambda \omega^{n-1};p)}\\
\times\prod_{j=2}^n{\theta(\omega u_1/u_j,\omega u_1/v_j;p)} \,Z_{n-1}(\omega u_2,\dots,\omega u_n; v_2,\dots, v_n;\lambda;p,\omega).\end{multline*}
Using  that $\theta(x,\omega x,\omega^2x;p)=\theta(x^3;p^3)$,
it follows that if $L_n(u;v)$ denotes the left-hand side of \eqref{zps}, then
\begin{multline*}L_n(u;v)\Big|_{u_1=p\omega v_1}\\
={(-1)^np^{-n}\lambda\omega^{2n+2}u_1^{2n-1}\theta(\omega;p)}
\prod_{j=2}^nu_jv_j{\theta(u_1^3/u_j^3,u_1^3/v_j^3;p^3)}L_{n-1}(\hat u;\hat v),\end{multline*}
where the hats indicate that $u_1$ and $v_1$ are omitted. This implies in turn
\begin{equation}\label{qr}Q_n(u;v)\Big|_{u_1=p\omega v_1}=\frac{(-1)^n\omega\theta(\omega;p)}{p^{1/2}\theta(-p^2;p^6)}\prod_{j=2}^nu_jv_j{\theta(u_1^3/u_j^3,u_1^3/v_j^3;p^3)} 
Q_{n-1}(\hat u;\hat v).\end{equation}
We need to show that $P_n^{(2)}$ satisfies the same recursion, when we specialize $z_1=z_{n+1}+1/3+\tau/2$.
With notation as in \eqref{psr}, in this specialization
$$P_n^{(2)}(z)=(-1)^{n+1}a_{1,n+1}\prod_{j\neq 1,n+1}b_{1,j}b_{n+1,j}\,P_{n-1}^{(2)}(\hat z). $$
After simplification, the coefficient can be identified with that in \eqref{qr}.
This completes the proof of Theorem \ref{dwt}.\end{proof}

Let us check that Theorem \ref{dwt} reduces to known results
 in the trigonometric limit $p\rightarrow 0$. It is easy to see that, in this limit,
$$\theta(-p^nx;p^2)\sim\begin{cases} p^{-n(n-2)/4}x^{-n/2}(1-x), & n \text{ even},\\
p^{-(n-1)^2/4}x^{(1-n)/2}, & n \text{\ odd}.\end{cases} $$
Using this together with Proposition \ref{tp} gives
\begin{multline}\label{td}
3^{\binom n2}V^n(1-\lambda\omega^{n+1})(1-\lambda\omega^{n+2})\,Z_n(\omega u;v;\lambda;0;\omega)\\
=\omega^{\binom{n+1}2}(U+\omega^n\lambda^2V)
s_{n-1,n-1,\dots,1,1,0,0}(u,v)\\
+(-1)^n\omega^{\binom{n}2}\lambda
s_{n,n-1,n-1,\dots,1,1,0}(u,v).
\end{multline}
The case $\lambda=0$ is \eqref{npfb}. For general $\lambda$ (that is, for the trigonometric 8VSOS model)
it follows from the proof of  \cite[Thm.\ 8.1]{r0} that
the left-hand side of \eqref{td} can alternatively be expressed as
\begin{multline}\label{tzd}\frac{\prod_{i,j=1}^n(u_i^3-v_j^3)}{\Delta(u,v)}\left(\omega^{\binom{n+1}2}(U+\omega^n\lambda^2V)\det_{1\leq i,j\leq n}\left(\frac{u_i-v_j}{u_i^3-v_j^3}\right)\right.\\
\left.+(-1)^n\omega^{\binom{n}2}\lambda\det_{1\leq i,j\leq n}\left(\frac{u_i^2-v_j^2}{u_i^3-v_j^3}\right)
\right). \end{multline}
As we have seen in \eqref{si}, the first terms in \eqref{td} and \eqref{tzd} can be identified.
That is also true for the second terms, cf.\ the second identity in \cite[Thm.\ 2.4(2)]{o}. 

\subsection{The three-colour model}

It was observed by Truong and Schotte \cite{ts} that the 8VSOS model contains the three-colour model as the special case when $q=u_i/v_j=\omega$ 
for all $i,j$. The latter model is obtained by starting from a height matrix and assigning to each square the weight $t_j$ 
if its height is equal to $j\ \operatorname{mod}\ 3$, where $t_0$, $t_1$ and $t_2$ are independent parameters. 
Let $Z_n^{\text{3C}}(t_0,t_1,t_2)$ denote the sum of the weight of all states with domain wall boundary conditions.
We will parametrize the weights as  
\begin{equation}\label{tj}t_j=\theta(\lambda\omega^j;p)^{-3}\end{equation} 
(since the partition function is homogeneous, two parameters are sufficient). An equivalent
parametrization was used by Baxter \cite{btc} to study the model for periodic boundary conditions.
By \cite[Eq.\ (4.8)]{r}, we then have
$$Z_n^{\text{3C}}(t_0,t_1,t_2)=\omega^{n(n+1)}\frac{\theta(\lambda\omega^2,\lambda\omega^{n+1};p)^2}{\theta(\lambda\omega^n;p)\theta(\lambda^3;p^3)^{n^2+2n+2}}\,Z_n(\omega,\dots,\omega;1,\dots,1;\lambda;p,\omega). $$
Applying Theorem \ref{hdt} and  Theorem \ref{dwt} yields the following new expression for the three-colour partition function $Z_n^{\text{3C}}$ in terms of Hankel determinants.

\begin{corollary}\label{thc}
In the parametrization \eqref{tj},
the domain wall partition function for the three-colour model can be expressed in terms of the Hankel determinants
\eqref{hd} as 
\begin{align*}
Z_n^{\text{\emph{3C}}}&=\frac{(-1)^{\binom{n+1}2}\omega^{2n^2}(p^3;p^3)_\infty^{3n^2}\lambda^n\theta(\lambda\omega^2,\lambda\omega^{n+1};p)^2}{\prod_{j=1}^{2n}(j-1)! (48p)^{\binom n2}(p;p)_\infty^{3n^2+1}(p^6;p^6)_\infty^{4n^2+1}\theta(\lambda^3;p^3)^{n^2+2n+3}}\\
&\quad\times\big(\omega^{2n}(-p;-p)_\infty(p^{12};p^{12})_\infty\lambda\theta(-p^{n+1}\omega^n\lambda^2;p^2)H_n^{(2)}
\\
&\qquad+(-p^3;-p^3)_\infty(p^4;p^4)_\infty\theta(-p^n\omega^n\lambda^2;p^2)H_n^{(4)}\big).
\end{align*}
\end{corollary}

In the derivation of Corollary \ref{thc}, we used that
\begin{align*}\Delta(x^2;p)&=\prod_{1\leq i<j\leq 2n}(e^{2\ti\pi z_i}-e^{2\ti\pi z_j})(px_i^2/x_j^2,px_j^2/x_i^2;p)_\infty\\
&\sim (2\ti\pi(p;p)_\infty^2)^{n(2n-1)}\Delta(z)\end{align*}
as all $z_j\rightarrow 0$, as well as some elementary manipulation of infinite products. In particular, the identity
$$(p^2;p^2)_\infty^3=(p;p)_\infty(-p;-p)_\infty(p^4;p^4)_\infty $$
implies that
$$ \theta(-p^2;p^6)C_2=-\theta(-p;p^6)C_4=\frac{(p^3;p^3)_\infty}{(p;p)_\infty(p^6;p^6)_\infty^2}.$$

Corollary \ref{thc} should be compared with \cite[Cor.\ 9.4]{r}, where we expressed the three-colour partition function in terms of certain polynomials $p_{n-1}$ in the variable
$$ \zeta=\frac{\omega^2\theta(-1,-p\omega;p^2)}{\theta(-p,-\omega;p^2)}.$$
In \cite{rb}, we showed that these polynomials can be interpreted as tau functions of Painlev\'e VI.
We can now deduce the following new Hankel determinant formula for $p_{n-1}$. 
It seems quite different in nature both from the Hankel determinants for general tau functions given in \cite{kn}
and for the expression for $p_{n-1}$ as an $(n-1)\times (n-1)$-determinant given in \cite[Lemma 9.7]{r}.

\begin{corollary}
The polynomial $p_{n-1}(\zeta)$ introduced  in \cite{r} has the determinant representation
$$p_{n-1}(\zeta)=A B^{[n^2/4]} H_n^{(2)}, $$
where
\begin{align*}A&=\begin{cases}\displaystyle\frac {1}{\prod_{j=1}^{2n}(j-1)!(48)^{\binom n2}p^{n(3n-2)/4}}, & n \text{ even},\\[5mm]
\displaystyle\frac {(-1)^{(n+1)/2}(p^3;p^3)_\infty^3(p^4;p^4)_\infty(p^{12};p^{12})_\infty}{\prod_{j=1}^{2n}(j-1)!(48)^{\binom n2}p^{(n-1)(3n+1)/4}(p;p)_\infty(p^2;p^2)_\infty^3(p^6;p^6)_\infty^5}, & n\text{ odd},\end{cases}\\
B&=\frac{(p^4;p^4)_\infty^6(p^3;p^3)_\infty^9}{(p;p)_\infty^3(p^2;p^2)_\infty^{15}(p^6;p^6)_\infty^{11}(p^{12};p^{12})_\infty^2}.
\end{align*}
\end{corollary}

In fact, a similar identity holds for each of our twelve pfaffians. The reason is that $\zeta$ is a Hauptmodul for $\Gamma_0(2,6)\simeq \Gamma_0(12)$. Hence, the normalizer $N$ discussed in \S \ref{mts} acts on $\zeta$ by rational transformations \cite{rsa}. 
It follows that, for any $\sigma\in \Sigma$, there exists a rational function $\phi_\sigma$ such that 
$p_{n-1}(\phi_\sigma(\zeta))$ is an elementary factor times $H_n^{(\sigma)}$. 
We have not worked out these relations in detail.

\subsection{Eight-vertex model on an odd chain}\label{ocs}

The eight-vertex model can  be parametrized by the Boltzmann
weights \cite{b0,bm1}
\begin{align*}
R^{++}_{++}&=R^{--}_{--}=\rho\,\theta_4(2\eta|2\tau)\theta_4(u-\eta|2\tau)\theta_1(u+\eta|2\tau),\\
R^{+-}_{+-}&=R^{-+}_{-+}=\rho\,\theta_4(2\eta|2\tau)\theta_1(u-\eta|2\tau)\theta_4(u+\eta|2\tau),\\
R^{+-}_{-+}&=R^{-+}_{+-}=\rho\,\theta_1(2\eta|2\tau)\theta_4(u-\eta|2\tau)\theta_4(u+\eta|2\tau),\\
R^{++}_{--}&=R^{--}_{++}=\rho\,\theta_1(2\eta|2\tau)\theta_1(u-\eta|2\tau)\theta_1(u+\eta|2\tau),
\end{align*}
where $u$, $\eta$ and $\tau$ are parameters of the model and
$\rho$ is a  normalization factor.  
We will only consider the case $\eta=\pi/3$, when the related XYZ spin chain is supersymmetric \cite{fe}. We take
$$\rho=\frac {p^{-1/4}}{(p^2;p^2)_\infty(p^4;p^4)_\infty},\qquad p=e^{\ti\pi\tau}.$$

If  $V$ is a vector space with basis $v_\pm$, one defines
$$R(u)(v_k\otimes v_l)=\sum_{m,n\in\{\pm\}}R_{kl}^{mn}v_m\otimes v_n. $$
Consider a tensor product $V_0\otimes V_1\otimes\dots\otimes V_N$, where each $V_j=V$. 
The inhomogeneous transfer matrix is 
the operator on $V_1\otimes\dots\otimes V_N$ given by 
$$\mathbf T(u)=\mathbf T(u;u_1,\dots,u_{N})=\operatorname{Tr}_0\left(R_{01}(u-u_1)\dotsm R_{0N}(u-u_{N})\right),$$
where $R_{ij}(u)$ denotes $R(u)$ acting on  $V_i\otimes V_j$. 

A $Q$-operator  is a family of operators $\mathbf Q(u)$ 
on $V^{\otimes N}$
such that $\mathbf Q(u)\mathbf T(v)=\mathbf T(v)\mathbf Q(u)$ and
\begin{equation}\label{tq}\mathbf T(u)\mathbf Q(u)=\phi(u-\eta)\mathbf Q(u+2\eta)+\phi(u+\eta)
\mathbf Q(u-2\eta), \end{equation}
where $\phi(u)=\prod_{j=1}^{N}\theta_1(u-u_j|\tau)$.

We now restrict to the case when $N$ is odd. In this case,
Razumov and Stroganov \cite{ras} conjectured that 
$\mathbf T(u)$ has an eigenvector with eigenvalue $\phi(u)$. Assuming that 
there exists a $Q$-operator $\mathbf Q(u)$ with the same eigenvector
and eigenvalue $Q(u)$, we must have
\begin{equation}\label{pq}\phi(u)Q(u)=\phi(u-\eta)Q(u+2\eta)+\phi(u+\eta)Q(u-2\eta). \end{equation}

The existence of a $Q$-operator is far from obvious.
Baxter gave two constructions of $Q$-operators for the eight-vertex model \cite{b0,bf}.
Unfortunately, the first of these does not work for the case $\eta=\pi/3$  \cite{fm}
and the second one does not work when $N$ is odd. 
 Bazhanov and Mangazeev suggest that the $Q$-operator from \cite{b0}  is well-defined
on a subspace containing the particular eigenvector that we are interested in. Alternative constructions of a $Q$-operator  
valid for $\eta=\pi/3$ and $N$ odd have been given by
by Fabricius \cite{f} and Roan \cite{ro}. However,  all the papers mentioned focus on the homogeneous chain ($u_1=\dots=u_{N}=0$).
We will not address the  problem of solving 
 the operator equation \eqref{tq}, but only consider the  scalar equation \eqref{pq}.

We will write $t=e^{i(u_1+\dots+u_{N})}$. Note that $\phi$ satisfies
\begin{equation}\label{phpx}\phi(u+\pi)=-\phi(u),\qquad \phi(u+\pi\tau)=-t^2e^{-\ti N(2u+\pi\tau)}\phi(u). \end{equation}
Let $Q$ be a solution to \eqref{pq}.
We will assume that $Q$ is an entire function and satisfies 
\begin{equation}\label{qqp} Q(u+2\pi)=Q(u),\qquad Q(u+2\pi\tau)=t^{-2}e^{-2\ti N(u+\pi\tau)}Q(u). \end{equation}
This agrees with the assumptions of \cite{bm1}, dealing with the homogeneous chain.
Using \eqref{phpx}, 
it is easy to see that the map $Q(u)\mapsto t^{-1}e^{\ti N(u+\pi\tau/2)}Q(u+\pi\tau)$ defines
an involution on the space of solutions to  \eqref{pq}  and \eqref{qqp}. Thus, any solution can be decomposed 
as $Q=Q_++Q_-$, where
$$ Q_\pm(u+2\pi)=Q_\pm(u),\qquad Q_\pm(u+\pi\tau)=\pm t^{-1}e^{-\ti N(u+\pi\tau/2)}Q_\pm(u). $$
We will now show that such solutions $Q_\pm$ exist and are unique up to normalization.
Let  $f(z)=\phi(2\pi z)Q_\pm(2\pi z)$. Using again \eqref{phpx}, one checks that
$$f(z+1)=f(z),\qquad f(z+\tau/2)=\mp \frac{t^3}{e^{3N\ti\pi(2z+\tau/2)}}\,f(z), $$
$$ f\left(z-\frac 23\right)+f(z)+f\left(z+\frac 23\right)=0. $$
Moreover, $f$ vanishes at the points $z_j$ and $z_j+1/2$, $j=1,\dots, N$, where $u_j=2\pi z_j$. 
This can be recognized as the conditions of Proposition \ref{ap}, with $\tau$ replaced by $\tau/2$, $N=2n-1$
and $\sigma=\hat 3, 3$ in the case of $Q_+$, $Q_-$, respectively.
Thus, we can deduce the following fact.

\begin{theorem}\label{tqt}
Suppose that $N=2n-1$ and $\eta=\pi/3$. 
Let $V$ be the space of solutions to the $TQ$-equation \eqref{pq}, which are entire and satisfy the quasi-periodicity
conditions \eqref{qqp}. Suppose that the parameters $u_1,\dots,u_{N}$ are generic. 
Then, $\dim V=2$ and $V$ is spanned by the two solutions
\begin{equation}\label{tqs}Q^{(\sigma)}(u)=\frac 1{\phi(u)}\,P_n^{(\sigma)}\left(\frac{u}{2\pi},\frac{u_1}{2\pi},\dots,\frac{u_{N}}{2\pi} ;\frac\tau 2\right),\qquad 
\sigma =3,\,\hat 3.\end{equation}
\end{theorem}

Note that, by Corollary \ref{hsc}, the two solutions are interchanged by the shift  $u\mapsto u+1/2$. 

For the homogeneous chain, Theorem \ref{tqt} leads to new determinant formulas for the  functions studied in \cite{bm1}.
To this end, we divide \eqref{tqs} by $\Delta(u_1,\dots,u_N)$ and then let $u_1,\dots,u_N\rightarrow 0$.
A slight variation of Lemma \ref{phl} then gives (up to an irrelevant constant factor)
$$Q^{(\sigma)}(u)\sim\frac{\theta_2(u/2|\tau/6)^N}{\theta_1(u|\tau)^N}\det\left [\begin{matrix}f(u)&f''(u)&\dotsm&f^{(N-1)}(u)\\f'(0) & f^{(3)}(0)&\dotsm& f^{(N)}(0)\\
f^{(3)}(0)& f^{(5)}(0) &\dotsm & f^{(N+2)}(0)\\ \vdots &&&\vdots\\ f^{(N-2)}(0) & f^{(N)}(0)&\dotsm & f^{(2N-3)}(0)\end{matrix}\right],$$
where
$$f(u)=\frac{\theta_1(u/2|\tau/2)\theta_2(u/2|\tau/2)\theta_l(u/2|\tau/2)}{\theta_2(u/2|\tau/6)} $$
and $l=3,\, 4$ if $\sigma=\hat 3,\,3$, respectively. This provides two linearly independent solutions to \eqref{pq} and \eqref{phpx}, with $u_1=\dots=u_N=0$. 
 
Finally,
we remark that elliptic pfaffians also appear in  Zinn-Justin's work  \cite{zj}
on the conjectured special eigenvectors of $\mathbf{T}(u)$; however, these are different in nature from those considered here.


 \end{document}